\documentclass{article}
\usepackage{authblk}
\usepackage[margin=0.9in,footskip=0.3in]{geometry}
\usepackage{longtable,float,hyperref,color,amsmath,amsxtra,amssymb,latexsym,amscd,amsthm,amsfonts,makecell}
\usepackage{booktabs}
\usepackage{multirow}
\usepackage{mathtools}
\usepackage{stmaryrd}
\allowdisplaybreaks
\usepackage{footnote}
\usepackage{tabularx,array}
\usepackage[numbers]{natbib}
\usepackage{enumitem}
\usepackage{hyperref}
\usepackage{graphicx, wrapfig, caption, adjustbox, fancyhdr}
\numberwithin{equation}{section}
\theoremstyle{plain}

\newcommand{\lb}{\llbracket}
\newcommand{\rb}{\rrbracket}
\newcommand{\lrbP}{\lb1,P\rb}
\newcommand{\lrbN}{\lb1,N\rb}

\newcommand{\cR}{\mathcal{R}}
\newcommand{\K}{\mathcal{K}}

\newcommand{\Omp}{\overset{\circ}{\Omega}}
\newcommand{\bomega}{\mathbf{\omega}}
\newcommand{\diag}{\mathrm{diag}}
\newcommand{\eps}{\varepsilon}

\newcommand{\N}{\mathbb{N}}
\newcommand{\R}{\mathbb{R}}
\newcommand{\bI}{\mathbf{I}}
\newcommand{\bD}{\mathbf{D}}
\newcommand{\bS}{\mathbf{D}}
\newcommand{\bu}{\mathbf{u}}
\newcommand{\bz}{\mathbf{z}}
\newcommand{\bun}{\mathbf{1}}
\newcommand{\bze}{\mathbf{0}}
\newcommand{\bv}{\mathbf{v}}
\newcommand{\bTheta}{\mathbf{\Theta}}
\newcommand{\q}{\pi}
\newcommand{\E}{\mathbb{E}}

\newtheorem{thm}{Theorem}[section]
\newtheorem{rmk}{Remark}[section]
\newtheorem{prop}[thm]{Proposition}
\newtheorem{cor}[thm]{Corollary}
\newtheorem{hyp}{Assumption}
\newtheorem{lemma}[thm]{Lemma}
\newtheorem{defi}{Definition}[section]
\title{Multi-strain SIS dynamics with coinfection under host population structure
}%
\author[1] {Sten Madec}
\author[2]{Nicola Cinardi}
\author[2]{Erida Gjini}
\affil[1]{Institut Denis Poisson, University of Tours, Tours, France}
\affil[2]{Center for Computational and Stochastic Mathematics, Instituto Superior Tecnico, Lisbon, Portugal}

\begin{document}
	
	\maketitle
	
	\begin{abstract}
		Coinfection phenomena are common in nature, yet there is a lack of analytical approaches for coinfection systems with a high number of circulating and interacting strains. In this paper, we investigated a coinfection SIS framework applied to $N$ strains, co-circulating in a structured host population. Adopting a general formulation for fixed host classes, defined by arbitrary epidemiological traits such as class-specific transmission rates, susceptibilities, clearance rates, etc., our model can be easily applied in different frameworks: for example, when different host species share the same pathogen, in classes of vaccinated or non-vaccinated hosts, or even in classes of hosts defined by the number of contacts. Using the strain similarity assumption, we identify the \textit{fast} and \textit{slow} variables of the epidemiological dynamics on the host population, linking neutral and non-neutral strain dynamics, and deriving a global replicator equation. This global replicator equation allows to explicitly predict coexistence dynamics from mutual invasibility coefficients among strains. The derived global pairwise invasion fitness matrix contains explicit traces of the underlying host population structure, and of its entanglement with the strain interaction and trait landscape. Our work thus enables a more comprehensive study and efficient simulation of multi-strain dynamics in endemic ecosystems, paving the way to deeper understanding of global persistence and selection forces, jointly shaped by pathogen and host diversity. 
	\end{abstract}

	\tableofcontents
	
	\setlength{\parindent}{0cm}
	
	\section{Introduction}
    
     In \cite{madec2020predicting}, we introduced an SIS coinfection model with multiple interacting strains, and presented an approach for simplification of its dynamics, based on slow-fast dynamics and strain similarity. Here, we develop the model extension away from the homogeneous mixing assumption, allowing transmission to be embedded in an explicit host population structure. Although the model is general, special cases include an explicit host contact network defined by the number of contacts, a host population defined by classes of different epidemiological parameters, or vaccination status. We derive the model reduction for this SIS epidemiological model with multi-strain coinfection on a structured host population, using the similarity assumption between strains. We do not impose any constraints on the type of structure considered, except that it has defined first and second moments.
	

	Essentially, in this paper, we are interested in understanding the $N-$ strain coinfection SIS model with several classes of hosts $\K$. The general ODE system for host proportions in different classes ($k$) and epidemiological status, following our previous works \cite{madec2020predicting,le2023quasi}, reads
		\begin{equation}\label{model:generalqskintro}
		\begin{cases}
			\dfrac{d}{dt}S_k=r_k(1-S_k)-S_k\sum\limits_{j=1}^N \beta_k^j\Theta_k^j +\sum\limits_{j=1}^N\gamma_k^j I_k^j+\sum\limits_{(j,l)\in\lrbN^2} \gamma_k^{j,l} D_k^{j,l}\\
			\dfrac{d}{dt}I_k^i=\beta_k^i \Theta_k^i S_k-(r_k+\gamma_k^i )I_k^i- I_k^i\sum\limits_{j=1}^N \sigma_k^{ij}\beta_k^j  \Theta_k^j\\
			\dfrac{d}{dt}D_k^{ij}=\sigma_k^{ij}\beta_k^j \Theta_k^j I_k^i-(r_k+\gamma_k^{i,j}) D_k^{ij} \\
		 
		\end{cases}
	\end{equation}
	where $\Theta_k^i$ denotes the total probability that the host class $k$ is infected with strain $i$ from any of the classes $s\in\K$, similar to mean-field approaches in network models \cite{moreno2002epidemic}.
\\	
Beyond the direct extension of our previous work \cite{madec2020predicting,le2023quasi}, our second objective here is also to address a big challenge in infectious disease modelling, related to the joint effect of several sources of heterogeneity. Here we combine: host population heterogeneity, contact rate heterogeneity, strain heterogeneity, and strain interactions, into the same model, including single and co-infection. How heterogeneity affects disease spread involves a large ongoing research, both from the theoretical and empirical perspectives \cite{ball1985deterministic,dwyer1997host, lloyd2005superspreading, andreasen2011final, gou2017heterogeneous, montalban2022herd,allard2023role,anderson2023quantifying,tuschhoff2025heterogeneity}. Yet, especially in the context of multi-strain infectious diseases with strain interactions and co-infection, it remains unclear how different layers of heterogeneity jointly shape both endemic global quantities and individual strain trajectories. With our model, we contribute to simplifying such a high-dimensional system, providing an explicit model reduction leading to a clearer understanding of aggregated mean-field quantities over the population, including mean-field strain dynamics, governed by emergent global fitness parameters. This enables an analytical study of the role of different sources of heterogeneity on the infectious disease dynamics and strain propagation.

The paper is organised as follows.  
In the next subsection, we define the notations used throughout the paper.  
In Section \ref{sec:2}, we present the general coinfection SIS model with host classes but without strain structure, and we recall the (essentially classical) results that this model satisfies. It is important to note that, even in the presence of coinfection, the results of this section are direct consequences of the classical SIS framework with host classes. Nevertheless, they are essential for the developments in the following sections.  

In Section \ref{sec:3}, we describe the general model with $N$ strains. After introducing our notion of neutrality in this context, we present our main result under the perfectly neutral assumption.  

In Section \ref{sec:4}, we introduce the notion of quasi-neutrality and state the main result of the paper: the Theorem \ref{Th:replicator}. This theorem shows that, under quasi-neutrality, the dynamics are governed by a replicator system for the strain frequencies $\bz$, similar to \cite{madec2020predicting,le2023quasi}. The parameters of this replicator system over the structured population depend directly on the deviation from neutrality among strains, and on the attractors of the corresponding no-strain system described in Section \ref{sec:2}.  

In Section \ref{sec:application}, we apply this general result to three different situations:  
(i) the case $N=2$, where all key quantities can be explicitly computed;  
(ii) a vaccination scenario, where host classes arise from heterogeneity in vaccine response;  
(iii) a mean-field network model, where host classes correspond to connectivity degrees in a network approach. The latter case is studied in detail in a companion paper focusing on the SIS dynamics on host contact networks \cite{networkPaper}.  

In these three applications, we show how the quasi-neutrality assumption enhances our understanding of strain interactions and ultimate strain selection in structured host populations. Section \ref{sec:five} provides some analytical results on the model applied to host contact networks. Section \ref{sec:discuss} summarizes our findings and conclusions. The most technical proofs are included in the final Section \ref{sec:proof}.
	Our results have implications for the analytical and computational study of strain selection dynamics in endemic ecosystems, shaped jointly by strain and host diversity.
	
	\subsection{Notations}
In this subsection, we present the notations used throughout the article and recall some basic properties. 
Most of the notations are straightforward, so the reader may choose to skip this section and return to it later if any notation becomes unclear during the flow of the main text.
	\paragraph{Matrices and Vectors.} 
Let \(n \in \mathbb{N}^*\) be a positive integer.  
We denote by \(\bze_n\) and \(\bun_n\) the column vectors of size \(n\) filled respectively with zeros and ones.  
We write \(\mathcal{M}_n\) for the set of all real square matrices of size \(n\).  
For a vector \(\bv = (v_1,\dots,v_n)\) of size \(n\), we denote by \(\mathrm{diag}(\bv) \in \mathcal{M}_n\) the diagonal matrix with entries of \(\bv\) on its diagonal. In particular, we set \(\mathbb{I}_n = \mathrm{diag}(\bun_n)\) for the identity matrix of \(\mathcal{M}_n\), and \(\mathbb{O}_n = \mathrm{diag}(\bze_n)\) for the zero matrix of \(\mathcal{M}_n\). 

A vector  $\mathbf{v}\in\R^n$ is say to be non-negative (reps. positive) and we denote  $\mathbf{v}\geq \mathbf{0}_n$ (resp. $\mathbf{v}>\mathbf{0}_n$
if for any $k$, $\mathbf{v}_k\geq 0 $ (resp. $\mathbf{v}_k>0$).

The same notations hold for matrix.

A matrix \(\mathbb{A} \in \mathcal{M}_n\) is called a {\em Metzler matrix} if all its off-diagonal entries are non-negative, i.e.:
\[\exists M\in\R,\; M\mathbb{I}_n+\mathbb{A}\geq 0.\]

\paragraph{Eingvalues}
The spectrum of a matrix \(\mathbb{A} \in \mathcal{M}_n\) is denoted by \(\mathrm{sp}(\mathbb{A})\).

We define the spectral radius
\[\rho(\mathbb{A})=\sup\{|\lambda|,\;\lambda\in\mathrm{sp}(\mathbb{A})\}\]
and the spectral bound 
\[\alpha(\mathbb{A})=\sup \{ \mathrm{Re}(\lambda),\;\lambda\in\mathrm{sp}(\mathbb{A})\}.\]

From the Perron Frobenius theorem, for any irreducible Metzler matrix \(\mathbb{A}\), \(\alpha(\mathbb{A})\) is the principal eigenvalue, and has the same properties as the spectral radius of a positive irreducible matrix. In particular, \(\alpha(\mathbb{A})\) is a simple eigenvalue of \(\mathbb{A}\), and the only one associated with a positive eigenvector see appendix \ref{Appendix:Metzler} for details.

	\paragraph{Subscripts and superscripts.}
Our model consists of three levels of structuring: host classes, strains, and infection levels.

	Host classes are indexed by a subscript $k\in\K$ where $\K$ is a finite set of size $p=|\K|$.
	For any variables $x_k$ labeled by $k\in\K$ we use the (slightly abusive) notation $x_\cdot=(x_{k_j})_{j\in\lrbP }=(x_k)_{k\in\K}$.
	- When context is clear, we simplify notation to 
	 $x=x_\cdot$
	
	Strains are indexed by superscripts $i\in\{1,\cdots,N\}=\lrbN$ that we write with an upper-script $i$.
	For any variables $x^i$ where $i\in\lrbN$ we denote $\mathbf{x}=(x^i)_{i\in\lrbN}$.
	
	Hence, in general $x_k^i$ represents an element corresponding to the strain $i\in\lrbN$ in the class $k\in\K$ and we have
	$x_\cdot^i=(x_k^i)_{k\in\K}$ and $\mathbf{x_k}=(x_k^i)_{i\in\lrbN}$.
	
	Depending on the context, the quantity $x_k^i$ may be a scalar or also a vector (of $\R^2$) representing infection stages (e.g., primo‑infected, secondly‑infected).
\paragraph{Statistics.}
Let $q=(q_k)_{k\in\K}$ be a probability vector, meaning that $q_k \geq 0$ for each $k \in \K$ and $\sum_{k\in\K} q_k = 1$.  
For any sequences $x=(x_k)_{k\in\K}$ and $y=(y_k)_{k\in\K}$, with a slight abuse of notation, we write
\[
\mathbb{E}_q(x_k) = \sum_{k\in\K} q_k x_k, 
\qquad 
\texttt{cov}_q(x_k,y_k) = \mathbb{E}_q(x_k y_k) - \mathbb{E}_q(x_k)\,\mathbb{E}_q(y_k).
\]

\paragraph{Kronecker product – matrix}
Recall that for two matrices $Q \in \mathcal{M}_p$ and $B \in \mathcal{M}_n$,  
the Kronecker product $M = Q \otimes B$ is the square matrix $M = (M_{ks})_{k,s} \in \mathcal{M}_{pn}$  
given by the $p \times p$ block structure
\[
M_{ks} = q_{ks} B \in \mathcal{M}_n.
\]
An important matrix in the following is
\[M=\text{diag}(B_{k_1},\cdots,B_{k_p}) \left(Q\otimes\mathbb{I}_2\right)\]
where $(B_k)_{k\in \K}$ is a list of $p$ matrix of $\mathcal{M}_2$ and $\text{diag}(B_{k_1},\cdots,B_{k_p})$ is the square  block diagonal matrix of $\mathcal{M}_{2p}$. 
$M\in\mathcal{M}_{2p}$ is given by the $p\times p$ block structure

\[
M_{ks} = q_{ks} B_k \in \mathcal{M}_2.
\]
	\paragraph{Kronecker product – vector}
	Let $N$ and $P$ be two natural numbers, and let $E$ be a vector space.  
	For any $\mathbf{x}=(x^i)_{i\in\lrbN} \in \mathbb{R}^N$ and $X=(X_k)_{k\in\lrbP}\in E^P$, we define the Kronecker product
	\[
	\mathbf{x}\otimes X = (x^i X)_{i\in\lrbN} = (x^i X_k)_{(i,k)\in\lrbN\times\lrbP} \in E^{N\times P}.
	\]
	If $\Omega \subset \mathbb{R}^N$, we denote
	\[
	\Omega \otimes \{X\} = \{\mathbf{x}\otimes X \;\mid\; \mathbf{x}\in \Omega\}.
	\]
	
	In our applications, a typical situation arises when $N$ is the number of species, 
	$P=|\mathcal{K}|$ is the number of classes, and $E=\mathbb{R}^2$ represents vectors of 
	primary and secondary infections.  
	The vector
	\[
	X^*=\left((I_k^*,D_k^*)^T\right)_{k\in\mathcal{K}} \in (\mathbb{R}^2)^P
	\]
	denotes a stationary pair in each class $k\in\mathcal{K}$, and the subset
	\[
	\Sigma^N=\Bigl\{\mathbf{z}\in[0,1]^N \;\Big|\; \sum_{i=1}^N z^i=1\Bigr\}
	\]
	is the standard probability simplex of $\mathbb{R}^N$. Hence,
	\[
	(u_k^i)_{(i,k)\in\lrbN\times\lrbP} \in \Sigma^N\otimes \{X^*\}
	\;\;\Longleftrightarrow\;\;
	\exists \mathbf{z}\in\Sigma^N \;\text{such that}\;
	u_k^i = z^i \begin{pmatrix} I_k^* \\ D_k^* \end{pmatrix}.
	\]
	
	This space represents all weighted combinations of stationary class-level pairs $(I_k^*, D_k^*)$, 
	with weights drawn from the probability simplex $\Sigma^N$.

	\section{The general model without strain structure}
	\label{sec:2}

We are interested in modeling co-infection across different classes of hosts.
Throughout this work, we make the following crucial hypothesis.

\begin{hyp}\label{hyp:constk}
	We assume that the total population of each class remains constant over time.
	Hence, without loss of generality, we may assume that this total abundance within each class is equal to 1.
\end{hyp}
In other words, in this text, we work with proportions of hosts within each class.  

Let $\K=\{k_1,\cdots,k_p\}$ denote the finite set of classes.  
For each class $k \in \K$, we denote the proportions of hosts in class $k$ which are susceptible, infected, and co-infected as $S_k$, $I_k$, and $D_k$, respectively.

	Let us denote $\Omega=\{(S,I,D)\in[0,1]^3,\; S+I+D=1\}^p$. 
	For any $t\geq 0$ and $k\in\K$, we write $X_k(t)=(S_k(t),I_k(t),D_k(t))$ and we write also  $X(t)=(X_k)_{k\in\K}\in \Omega$

The model is built upon a classical SIS model within class with the addition of co-infection.

In each class, we denote $r_k$ the birth and death rate and $\gamma_k$ the clearance rate.
	
The class $k$ is infected by the infected individual from the class $s$ at a rate $\beta_{ks}$. 
	
	Denoting $\beta_k=\sum_{s\in\K} \beta_{ks}^i$ and $q_{ks} =\dfrac{\beta_{ks}}{\beta_k}$ yields
	\[\beta_{ks}=\beta_k q_{ks}\text{ with }\sum_{s\in\K} q_{ks}=1.\]
	The number $q_{ks}$ may be interpreted as the probability for the class $k$ to be in contact with the class $s$.
	
	Hence, the probability that the class $k$ is infected by the class $s$ is $q_{ks}(I_s+D_s)$.
	Finally,  the probability that the class $k$ is  infected is
	
	\[\Theta_k=\sum_{s} q_{ks} (I_s+D_s).\]
	The matrix $Q=(q_{ks})_{(k,s)\in\K^2}$ is a non-negative matrix describing the connectivity between the classes. Moreover we assume 
	\begin{hyp}
		The non-negative connectivity matrix $Q$ is irreducible.
				Hence, $Q^T$ is an irreducible Markov matrix.
	\end{hyp}
	Once infected, the primary infection increases with a rate $\beta_k$. The rate of primary infection is then $\beta_k\Theta_k$.
	
	The second infection increases with a rate $\sigma_k\beta_k$. 
	 The rate of second infection is then $\sigma_k\beta_k\Theta_k$.
	
	These parameters $\sigma_k$ are specific to this coinfection model and represent the ratio between the (rate of) first infection and the (rate of) second infection in the class $k$.

	The system reads for each $t>0$:

	\begin{equation}\label{SIDS}
		\begin{cases}
			\dfrac{d}{dt}{S}_k=r_k(1 - S_k)-\beta_k S_k \Theta_k +\gamma_k (I_k +D_k )\\
			\dfrac{d}{dt}{I}_k= \beta_k S_k \Theta_k -\sigma_k \beta_k 
			I_k \Theta_k  -(\gamma_k+r_k) I_k \\
			\dfrac{d}{dt}{D}_k= \sigma_k\beta_k 
			I_k \Theta_k  -(r_k+\gamma_k) D_k. \\
		\end{cases}
	\end{equation}

	Let us denote $E_0=\{(1,0,0)\}^p\subset \Omega $ the (singleton of the) disease-free equilibrium and denote also $\Omp=\{(S,I,D)\in\Omega,\;S<1\}^p$.

It is clear that both $\Omega$, $\Omp$ and $E_0$ are invariant under the SIDS system \eqref{SIDS}. Moreover, the following result is standard	
	\begin{prop} Assume that $Q$ is irreducible.
	
	If $(S_k(0),I_k(0),D_k(0))_{k\in\K}\in \Omega\setminus E_0$ then $X(t)\in \Omp$ for each $t>0$.
	
\end{prop}

	  In other words, if each class can be infected by others - possibly in several steps - then the presence of the disease in any class leads to its immediate spread across all classes.

	An important feature of this model is that, denoting $T_k=I_k+D_k=1-S_k$, we find the classical $SIS$ model with class of hosts.
	\begin{equation}\label{SIS}
		\begin{cases}
			\dfrac{d}{dt}{S}_k=r_k(1 - S_k)-\beta_k S_k \Theta_k +\gamma_k T_k\\
			\dfrac{d}{dt}{T}_k= \beta_k S_k \Theta_k  -(\gamma_k+r_k) T_k \\
		\end{cases}
	\end{equation}
	with $\Theta_k=\sum_{s\in\K} q_{ks} T_s=\mathbb{E}_{q_{k\cdot}} (T)$, that is, $\Theta_\cdot=(\Theta_k)_{k\in\K}$ and $T_\cdot=(T_k)_{k\in\K}$ are related through the linear relation
	\[\Theta_\cdot= Q T_\cdot.\]

\begin{rmk}[An important particular case]
	A particular case of interest arises when the probability for a class $k$ to be in contact with class $s$ is independent of $k$. In this situation, we denote for simplicity $q_{ks}=q_s$ for any $(k,s)\in\K^2$ and we have
	$Q=\bun q^T$.
	It comes
	\[
	\Theta=\sum_{s\in\K} q_s T_s=\mathbb{E}_q(T).
	\]
	
	This scalar $\Theta\in[0,1]$ drives the infection throughout the whole system.
	
	This setting simplifies several interpretations, results, and even some proofs. Throughout this text, we will often restate the results in this particular situation.
\end{rmk}

	In our model, we divided the infected host $T_k$ into single-infected hosts $I_k$ and co-infected hosts $D_k$. Consequently, the classical results from SIS models with classes of hosts can be directly applied.  
	
		The dynamics of this model is well known from the pioneering paper \cite{SISk76} and has been revisited in \cite{SISkLyap2007} using Lyapunov function and in \cite{DDZ2022} within a framework of a continuum of classes $\K$.

	The result is the following. Denote  $Q=\left(q_{ks}\right)_{(k,s)\in\K^2}$, $\mathcal{R}_k=\dfrac{\beta_k}{r_k+\gamma_k}$ and $\mathcal{R}=(\mathcal{R}_k)_{k\in\K}$. With these notations, the next generation matrix is $\text{diag}(\mathcal{R} )Q=\left(\mathcal{R}_k q_{ks}\right)_{(k,s)\in\K^2}$. The threshold for the persistence of the disease is given by the spectral radius of this next-generation matrix : 
	\begin{equation}\label{def:R0}\mathcal{R}_0=\rho(\text{diag}(\mathcal{R} )Q ).\end{equation}
	
	As expected, $\mathcal{R}_0$ acts as a threshold for the endemic persistence of the disease. The full analysis follows in the next subsection.
\begin{rmk}\label{rmk:pi}
	In general, since the next generation matrix is a positive irreducible matrix, the Perron-Frobenius theorem implies that $\mathcal{R}_0$ is the principal eigenvalue of this matrix, this eigenvalue has multiplicity one, and this is the only eigenvalue such that the (right and left) eigenvectors have positive entries.
	
	In particular, up to a multiplicative constant, there exists a unique vector $\zeta=(\zeta_k)_{k\in K}$ and a vector $\mathfrak{I}=(\mathfrak{I}_k)_{k\in K}$ such that  $\zeta_k>0$ and $\mathfrak{I}_k>0$ for each $k\in\K$ and 
	\[\zeta^T\text{diag}(\mathcal{R} )Q =\mathcal{R}_0 \zeta^T\text{ and } \text{diag}(\mathcal{R} )Q \mathfrak{I}=\mathcal{R}_0 \mathfrak{I}. \]  
	
	In the particular case $Q=(q_s)_{(s,k)\in\K^2}=\mathbf{1} q^T $ then we have 
	\[q^T\text{diag}(\mathcal{R} )Q=q^T\text{diag}(\mathcal{R} )\mathbf{1} q^T=\left(\sum_{s\in\K} q_s \mathcal{R}_s\right) q^T\]
	so $q^T=\zeta^T$ is a positive (left) eigenvector of the positive irreducible matrix $\text{diag}(\mathcal{R})Q$ and then 
	
	\begin{equation}\mathcal{R}_0=    \sum_{s\in\K}{ q_s \mathcal{R}_s}=\mathbb{E}_q\left(\mathcal{R}_k\right).\end{equation}
    \end{rmk}

    \subsection{Analysis of the equilibria}
	\begin{prop}\label{prop:E0stable}[Stability of the DSE] Let $E_0=(\bun,\bze,\bze)$ be the Disease Free Equilibium (DSE).\\
		\begin{itemize} 
		\item If $\mathcal{R}_0>1$ then $E_0$ is unstable.
		\item If $\mathcal{R}_0<1$ then $E_0$ is linearly stable.
		\end{itemize}
		\end{prop}
				\begin{proof} We start to rewrite the system with the  variables 
				$\left(\Sigma_\cdot,T_\cdot,I_\cdot\right)$ where $\Sigma_k=S_k+I_k+D_k$ and $T_k=I_k+D_k$. With these variables, the system  \eqref{SIDS} reads equivalently as
				\begin{equation}\label{eq:SIDrewritten}\begin{cases}
					\dfrac{d}{dt}{\Sigma_k}=r_k(1-\Sigma_k),\;k\in\K\\
					\dfrac{d}{dt}{T_k}=\beta_k (\Sigma_k-T_k) \Theta_k-(r_k+\gamma_k) T_k,\;k\in\K\\
					\dfrac{d}{dt}{I_k}=\beta_k (\Sigma_k-T_k) \Theta_k -\sigma_k \beta_k I_k \Theta_k-(\gamma_k +r_k) I_k,\;k\in\K\\

				\end{cases}\end{equation}
				where $\Theta_k=\sum_{s\in\K} q_{ks} T_s$

				The steady state $E_0$ reads $(\Sigma_k,T_k,I_k)=(1,0,0)$, $\forall k\in\K$.

				This choice allows us to write the Jacobian matrix at $E_0$ in a block structure, where each block is a square matrix of size $|\K|\times |\K|$:
				\[Jac=\begin{pmatrix}
					-\text{diag}(r)&\mathbb{O}&\mathbb{O}\\
					\mathbb{O}&\mathbb{J}&\mathbb{O}\\
					\mathbb{O}&\mathbb{B}&-\text{diag}(\gamma+r)\\
				\end{pmatrix}\]
				with $\mathbb{O}=\mathbb{O}_{|\K|}$, $\mathbb{B}=\text{diag}(\beta) Q=(\beta_k q_{ks})_{(k,s)\in\K^2}$ and 
				$\mathbb{J}=\mathbb{B}-\text{diag}(\gamma+r).$
				
				It follows that $E_0$ is linearly stable if and only if the spectrum of $\mathbb{J}$ lies entirely in the left half-plane.
				3
				$\mathbb{J}$ is a Metzler matrix and we may apply theorems \ref{Th:A1} and \ref{Th:A2}.
				In particular, the spectrum of $\mathbb{J}$ lies entirely in the left half-plane if and only if its principal eigenvalue  $\alpha(\mathbb{J})$ is negative.
				
				Recall the notations  $\mathcal{R}_k=\dfrac{\beta_k}{r_k+\gamma_k}>0$ and  $\mathcal{R}=\left(\mathcal{R}_k\right)_{k\in\K}$.
				Since $\mathcal{R}_0=\rho \left(\text{diag}(\mathcal{R})Q\right)$, by the Perron-Frobenius theorem, there exists $\mathfrak{T}\in (0,+\infty)^{p} $ such that 
				
				\[\text{diag}(\mathcal{R})Q \mathfrak{T}=\mathcal{R}_0\mathfrak{T}. \]
				Multiplying by $\text{diag}(\gamma+r)$ and recombining we deduce
				
				\[\mathbb{J}\mathfrak{T}=\mathbb{B} \mathfrak{T} -\text{diag}(\gamma+r)\mathfrak{T}=(\mathcal{R}_0-1)\diag(\gamma+r)\mathfrak{T}\]
				
				If $\mathcal{R}_0-1<0$ then by  theorem \ref{Th:A2}-(iii) we have $\alpha(\mathbb{J})<0$.
				
				If $\mathcal{R}_0-1>0$ then by theorem \ref{Th:A2}-(v) we have $\alpha(\mathbb{J})> 0$ which ends the proof.
			\end{proof}
			Classically, when $\mathcal{R}_0>1$ there exists (a unique) positive steady state $E^*$ which is linearly stable.
			\begin{prop}\label{prop:steadystate}
					Assume that $\mathcal{R}_0=\rho(\text{diag}(\mathcal{R})Q )>1$.
 There exist a unique endemic equilibrium $E^*=(S_k^*,I_k^*,D_k^*)_{k\in\K}\in\Omp$. Moreover, $E^*$ is linearly stable.
			\end{prop}
			\begin{proof}
				
				{\em (i) Existence and uniqueness.} 	This result is well known, see for instance \cite{SISkLyap2007}.
                
                Here,  we propose an independent proof in the particular  case $Q=\bun q^T.$\\
                Assume that there exists a steady state in $\Omp$. Then clearly $\Theta=\E_q(T_k)\in(0,1)$. From \eqref{SIS} we have for each $k\in\K$:
				\begin{equation}\label{Tkstar}
					T_k=\mathcal{R}_k (1-T_k)\Theta.
				\end{equation}
				Thus
				\[T_k=\dfrac{\Theta \mathcal{R}_k}{1+\Theta\mathcal{R}_k}.\]
				Since $\Theta=\E_q(T_k)$, multiplying \eqref{Tkstar} by $q_k$ and summing over $k\in K$ yields
				$\Theta=\Theta\sum_{k\in\K} q_k\mathcal{R}_k  (1-T_k).$
				Since $\Theta>0$ and writing $(1-T_k)=\dfrac{1}{1+\Theta \mathcal{R}_k}$,  this implies the implicit equation on $\Theta$:
				\begin{equation}\label{implicitTheta}
					1=\sum_{k\in\K} q_k\dfrac{\mathcal{R}_k}{1+\Theta \cR_k} :=F(\Theta)
				\end{equation}			
				The function $F$ is decreasing and satisfies $F(1)=\E_q \left(\dfrac{\cR_k}{1+\cR_k}\right)<1$. Therefore, it exists $\Theta\in(0,1)$ such that $F(\Theta)=1$ if and only if $F(0)=\E_q (\cR_k)=\cR_0>1$ and if such a $\Theta$ exists, it is unique. 
				
				Now assume that $\cR_0>1$ and  denote  $\Theta^*$ this value.\\
				The explicit values of  $T_k^*$ and  $S_k^*=1-T_k^*$ is given above in the proof and the values of $I_k^*$ and $D_k^*$ follows directly from \eqref{SIDS}.\\
			{\em (ii) Stability.} We prove it in general.

            As above, in the proof of the previous proposition, we rewrite equivalently the system for the variable $(\Sigma_k,T_k,I_k)_{k\in\K}$ obtaining  \eqref{eq:SIDrewritten}. The steady state  $E^*$ reads  $(\Sigma_k,T_k,I_k)=(1,T_k^*,I_k^*),\; \forall k\in\K$. We denote also $\Theta^*=QT^*$.The Jacobian matrix at $E^*$ reads 

            \[Jac=\begin{pmatrix}
					-\text{diag}(r)&\mathbb{O}&\mathbb{O}\\
					\textrm{diag}(\beta \Theta^*)&\mathbb{J}&\mathbb{O}\\
					\textrm{diag}(\beta \Theta^*)&\mathbb{B}&-\mathrm{diag}\left(\gamma+r+\sigma  \beta\Theta^*\right)\\
				\end{pmatrix}.\]
                With $\mathbb{B}=-\mathrm{diag}(\beta T^*+\sigma \beta I^*) Q$ and $\mathbb{J}=-\mathrm{diag}(\beta \Theta^*+r+\gamma)+\mathrm{diag}(\beta (1-T^*)) Q$.

            From the block triangular structure,  we see that $\alpha(Jac)<0$ if and only if $\alpha(\mathbb{J})<0.$

            Remark that $\mathbb{J}$ is a Metzler matrix. 

            From the identities $\diag(r+\gamma)\diag(\mathcal{R})=\diag(\beta)$ and  $T^*=\diag(\mathcal{R} (1-T^*)) QT^*$, we get
        \[\mathbb{J}T^*=-\diag(r+\gamma)\left(\diag(\mathcal{R}\Theta^*+1)T^*+\diag\left(\\\mathcal{R} (1-T^*)\right)QT^*\right)=-\diag(r+\gamma)\left(\diag(\mathcal{R}\Theta^*)\right)T^*<0.\]
        The conclusion follows from Theorem \ref{Th:A2}-(iii).
			\end{proof}
		Finally, like in the standard SIS, we can completely describe the dynamics. 
	\begin{thm}\label{th:SISqsk}
		Let $\mathcal{R}_0=\rho(\text{diag}(\mathcal{R})Q )$.
		\begin{itemize}
			\item If $\mathcal{R}_0\leq 1$ then the DSE  $E_0$ is the only equilibria and is globally stable in $\Omega$.
			\item If $\mathcal{R}_0>1$ then there exist a unique endemic equilibrium $(S_k^*,I_k^*,D_k^*)_{k\in\K}$ which is globally stable in $\Omega\setminus \{E_0\}$.
			Moreover, we have:
			\[S_k^*=\dfrac{1}{1+\mathcal{R}_k \Theta_k^*},\;T_k^*=\dfrac{\mathcal{R}_k \Theta_k^*}{1+\mathcal{R}_k\Theta_k^*},\; I_k^*=\dfrac{T_k^*}{1+\sigma_k\mathcal{R}_k \Theta_k^*}\text{ and }D_k^*=\dfrac{\sigma_k \mathcal{R}_k\Theta_k^*T_k^*}{1+\sigma_k\mathcal{R}_k \Theta_k^*}.\]
			Lastly, the convergence is exponentially fast :
			there exists $\eta>0$ and $C>0$ such that 
			
			\[|S_k(t)-S_k^*|+|I_k(t)-I_k^*|+|D_k(t)-D_k^*|\le C e^{-\eta t}.\]
		\end{itemize}
		
	\end{thm}
\begin{proof}
The proof of global convergence for the SIS model can be found in \cite{SISkLyap2007}, using Lyapunov functions, and in \cite{DDZ2022}, which relies directly on the monotone structure of the system.

The extension to the coinfection model is straightforward.

When $\mathcal{R}_0>1$, once we know that there exists a unique $\Theta^*=(\Theta^*_k)_{k\in\K}$ at the endemic equilibrium, the formulas for $S_k^*, I_k^*, D_k^*$ follow directly.

Finally, exponential convergence follows from global convergence together with the linear stability of the endemic equilibrium.
\end{proof}
	We state an easy consequence that is useful for the next section.
	\begin{cor}\label{corxi}
		Let $\mathcal{R}_0=\rho(\text{diag}(\mathcal{R})Q )$ and assume that $\mathcal{R}_0>1$. We have
		
		\[\text{diag} (\mathcal{R} S^*) Q T^*=T^*\]
		
		Thus, $\rho\left(\text{diag} (\mathcal{R}S^*) Q\right)=1$, and therefore there exists a unique vector $\q=(\q_k)_{k\in\K}$ satisfying $\sum_{k\in\K}\q_k=1$ and $\q_k>0$ for each $k\in\K$, such that
		\[\q^T\text{diag} (\mathcal{R} S^*) Q=\q^T.\]
		
		In particular, if $q_{ks}=q_s$ is independent of $k$, that is $Q=q^T \bun$, then 
		$\q=q.$
	\end{cor}
	
	\begin{proof}
		By Theorem \ref{th:SISqsk}, if $\mathcal{R}_0>1$ then the system \eqref{SIS} admits a steady state $(S_k^*,T_k^*)_{k\in\K}$. The second equation reads
		\[\forall k\in\K,\; T_k^*= \mathcal{R}_k  S_k^* \Theta_k^*=  \mathcal{R}_k  S_k^* (QT^*)_k.\]
		
		In other words, we have
		\[\big[\text{diag} (\mathcal{R} S^*) Q \big]T^*=T^*.\]
		
		Hence $T^*$ is an eigenvector of $\text{diag} (\mathcal{R} S^*) Q$ associated with the eigenvalue $1$.  
		Since $\text{diag} (\mathcal{R} S^*) Q$ is a nonnegative and irreducible matrix, the remaining results are direct applications of the Perron–Frobenius theorem.
	\end{proof} 
	
	\begin{rmk}
	If $Q=\bun q^T$, and if  $\mathcal{R}_0=\mathbb{E}_q (\mathcal{R}_s)>1$ then the formula of theorem \ref{th:SISqsk} holds true with $\Theta_k^*=\Theta^*$ for each $k\in\K$ where $\Theta^*$ is characterized by 
	\begin{equation}\label{eq:theta}
		1=\sum_{k\in\K} \dfrac{q_k \mathcal{R}_k}{1+\mathcal{R}_k\Theta^*}.
		\end{equation}
\end{rmk}

	\paragraph{Motivations.} This general host structure can be applied in various situations where hosts exhibit specific characteristics and co-infection plays a significant role. The main assumptions are: $(i)$ within each class, the host population remains constant; and
	$(ii)$ there is no flux between host classes.
	In Section \ref{sec:application}, we explore three specific applications: (1) A case involving two markedly distinct host classes — such as two species — where all expressions are explicit. (2) A vaccination scenario, where host classes represent different responses to the vaccine. (3) A heterogeneous network model under a mean-field transmission assumption.
    
	\paragraph{Link with the literature} The model \eqref{SIDS} has been extensively studied since the seminal paper by \cite{SISk76}. Essentially, a strong property of the system is that it is monotone \cite{smith1995monotone}. This can be used to show that the situation is similar to the single-class situation. More precisely, there is a key positive quantity $\mathcal{R}_0$ such that if $\mathcal{R}_0\leq 1$ the disease-free equilibrium is globally stable and if $\mathcal{R}_0>1$, then there exists a single positive equilibrium - an endemic state -  which is globally stable.  The reader may see \cite{SISkLyap2007} for a more recent proof of these well-known results using Lyapunov function. 
	This model was most recently extended into a unified framework—including a continuum of classes—in \cite{DDZ2022}, where monotonicity arguments are extensively used.
Let us finally mention the extension proposed in \cite{Horder2022}, where higher-order interactions between classes are shown to create the possibility of bistability when $\mathcal{R}_0<1$.
 In the present work, we introduce a second infection level, denoted $D_k$ (coinfected host compartments), which does not alter the classical results, since the aggregated variables $(S_k,I_k+D_k)$ remain governed by Eq. \eqref{SIS}.
	 The key novelty of this paper lies in the incorporation of a strain structure atop this foundational $SIDS$ framework.
	
%

\subsection{Heterogeneity and endemic prevalence of infection}
Our first result concerns the effect of the heterogeneity of the $\mathcal{R}_k$ on the global probability of infection.
\begin{prop}[Effect of heterogeneity on the prevalence $T^*$]
\label{prop:prevalence}
Assume that $\mathcal{R}_0 > 1$ and let $\zeta = (\zeta_k)_{k \in \K}$ be the positive right eigenfunction introduced in Remark~\ref{rmk:pi}, normalized so that $\zeta^T \bun = 1$.

Note that $\zeta$ is a probability vector, and recall the notation $\mathbb{E}_\zeta(V) = \zeta^T V$.

Let $T^0 = 1 - \dfrac{1}{\mathcal{R}_0}$ be the limit prevalence of infection in the single-class (homogeneous) SIS system with basic reproduction number $\mathcal{R}_0$.

Then we have
\[
    \mathbb{E}_\zeta(T^*) \le T^0,
\]
with equality if and only if $\mathcal{R}_k = \mathcal{R}_0$ for each $k \in \K$.
\end{prop}

\begin{rmk}
    The vector $\Theta = (\Theta_k)_{k \in \K}$ represents the probabilities of infection. It is easy to construct examples where $\Theta_k > \Theta^0$ for some $k \in \K$, but this proposition shows that the mean probability of infection $\mathbb{E}_\zeta(\Theta^*)$ is always smaller than the probability of infection in the homogeneous system with the same $\mathcal{R}_0$.
\end{rmk}
    \begin{proof}
        We start from the identity for each $k\in\K$:
        \[T_k^*=\mathcal{R}_k (1-T_k^*)\Theta_k^*. \]
        Dividing by $1-T_k^*$, multiplying by $\zeta^T$, using $QT^*=\Theta^*$ and the definition of $\mathcal{R}_0$, we get
        \[\mathcal{R}_0 \mathbb{E}_\zeta( T^*)= \mathbb{E}_\pi \left(\dfrac{T_k^*}{1-T_k^*}\right)\]
   
        The function $x\mapsto x/(1-x)$ is convex on $(0,1)$ so by the Jensen inequality:
        \[\mathcal{R}_0\mathbb{E}_\pi (T^*)\geq  \dfrac{\mathbb{E}_\zeta(T_k^*)}{1-\mathbb{E}_\pi(T_k^*)}\]
        which can be rewritten as

        \[1-\dfrac{1}{\mathcal{R}_0}\geq \mathbb{E}_\zeta \left(T_k^*\right) \]
        
    \end{proof}
In the particular case where $Q = q^T \bun$, this result is stronger, since in this situation $\Theta_k^* = \Theta^*$ for all $k \in \K$. The result then becomes a statement about the effect of heterogeneity on the probability of infection in each patch.

\begin{cor}
    \label{prop:SIDS}
    Assume that $Q = q^T \bun$. Assume that $\mathcal{R}_0 > 1$ and denote by $T_k^*$ and $\Theta^*$ the quantities defined in Theorem~\ref{th:SISqsk}.
    
    Denote by $T^0 = \Theta^0 = 1 - \dfrac{1}{\mathcal{R}_0}$ the endemic equilibrium prevalence of the system with a single class sharing the same $\mathcal{R}_0$.
    
    Then \[\Theta^* \leq \Theta^0,\] with equality if and only if $\forall k \in \K, \quad \mathcal{R}_k = \mathcal{R}_0$.
    
\end{cor}

\begin{proof}
    This is a direct consequence of Proposition~\ref{prop:prevalence}, since $\pi^T = q^T$ in this case. However, in this special case, it is possible to provide an independent proof yielding a more explicit estimate.

	Recall that $\mathcal{R}_0=\E_q\left( \mathcal{R}_k\right)$.
		Assume that there exists $k\in\K$ such that $q_k>0$ and $\mathcal{R}_k\neq \mathcal{R}_0$ (otherwise the proof is trivial).
		

		The definition of $\Theta^*=\E_q\left(T_k^*\right)$ and $\Theta^0$ implies 
		
		\[\sum_{k\in\K} q_k \dfrac{\mathcal{R}_k}{1+\Theta^* \mathcal{R}_k}=1=\dfrac{\mathcal{R}_0}{1+\Theta^0 \mathcal{R}_0}=\sum_{k\in\K} q_k \dfrac{\mathcal{R}_k}{1+\Theta^0 \mathcal{R}_0}.\]
		Subtracting the left-hand side and dividing by $\dfrac{\Theta^0 \mathcal{R}_0}{1+\Theta^0 \mathcal{R}_0}$ yields
		\[\sum_{k\in\K} q_k \dfrac{\mathcal{R}_{k}}{1+\Theta^* \mathcal{R}_k} \left(1-\dfrac{\Theta^* \mathcal{R}_k}{\Theta^0 \mathcal{R}_0}\right)=0.\]
		
		From this, we infer 
		\begin{equation}\label{Theta0}\dfrac{\Theta^0 }{\Theta^* }=\sum_{k\in\K} q_k \dfrac{\mathcal{R}_{k}}{1+\Theta^* \mathcal{R}_k}\dfrac{\mathcal{R}_k}{\mathcal{R}_0}.\end{equation}

		Let $F\;:\; r\mapsto \dfrac{r}{r+r\Theta^* }$. 
		From the definition of $\Theta^*$ and $\mathcal{R}_0$, we have respectively $\E_q\left(F\left(\mathcal{R}_k\right)\right)=1$ and $\E_q\left(\dfrac{\mathcal{R}_k}{\mathcal{R}_0}\right)=1$.
		Thus, with these notations, \eqref{Theta0} reads
		\[\dfrac{\Theta^0 }{\Theta^* }=\mathtt{cov}_q\left(F\left(\mathcal{R}_k\right),\dfrac{\mathcal{R}_k}{\mathcal{R}_0}\right)+1. \]
		
		Since $F$ is increasing, we have $\mathtt{cov}_q\left(F\left(\mathcal{R}_k\right),\dfrac{\mathcal{R}_k}{\mathcal{R}_0}\right)>0$ which ends the proof. 
        
\end{proof}
This result proves that endemic prevalence in a heterogeneous host population is lower than that in a homogeneous population with the same $R_0.$

	\section{The $N-$ strain model}
	\label{sec:3}
	\subsection{Model description}
	Now we describe the incorporation of $N$ strains (or species) in the model. 
	Following  \cite{le2023quasi} and \cite{madec2020predicting}, the system is written to account for transmission history.
	Hence, for each class of host $k\in\K$, we define: 
	\begin{itemize}
		\item $I_k^i$ the proportion of the  single-infected by the strain $i\in\lrbN$ in the class $k$ and $\bI_k=(I_k^i)_{i\in\lrbN}$.
		\item  $D_k^{i,j}$ the proportion of the  double-infected hosts, first by the strain $i\in\lrbN$ and then by $j\in\lrbN$ in the class $k$ and 
		$\bD_k~=~(D_k^{i,j})_{(i,j)\in \lrbN^2}$.
	\end{itemize}
	Thus \[\forall k\in\K,\;S_k+\sum_{1\leq i\leq N} I_k^i+\sum_{1\leq i,j\leq N} D^{i,j}_k =1.\]	
	which reads shortly noting $\bun_q^T=(1,\cdots,1)\in\R^q$:
	\begin{equation}\label{kdensity}
		\forall k\in\K,\;S_k+\bun_N^T  \bI_k+\bun_{N^2}^T \bD_k=1
	\end{equation}

	\begin{rmk}[\underline{On the classes $D_k^{i,i}$ and the Neutral null property}]
		An important feature of this model is that it takes into account the classes $D^{i,i}$ of hosts which are infected twice by the {\it same} strain (see also \cite{van1995dynamics}). As it is well explained in \cite{Lipsitch2009} and discussed in detail in \cite{alizon_co-infection_2013} (in models free of additional population structure), this coinfection class is necessary for the model to be well-posed in that it satisfies the Neutral null property :  
		
		{"If all strains are identical, then there is no structural advantage coming from the model for any of them." }
		
		In particular, as it is used below, if every strain is completely equivalent to any other strain, then summing all of them yields exactly the SID model \eqref{SIDS}. This important property is not true if this class is omitted.
	\end{rmk}
	
	Denote $\Upsilon_N=\{(S,\bI,\bD)\in[0,1]\times [0,1]^N\times [0,1]^{N^2},\text{s.t. \eqref{kdensity} holds}\}$
	and $\Omega_N=\Upsilon_N^p$. A state of our system at a time $t\in\R$ is then $X(t)=(S_k,\bI_k,\bD_k)_{k\in\K}\in \Omega_N$. We note also $\overset{\circ}{\Omega}_N=\{(S_k,\bI_k,\bD_k)_{k\in\K}\in \Omega_N,\;S_k<1,\;\forall k \in \K \}$ and the disease free equilibrium $E_{0}=\{(1,\bze_N,\bze_{N^2})\in\Upsilon_N\}^p$.\\

	For the general system, we allow every parameter to be strain-dependent (see table \ref{tabledef}). Most of the parameters are straightforward.  Let us explain the infection processes embedded within the model.
	
	If a host in the class $I_k^i$ transmits the disease, it transmits the strain $i$.\\
	If a host in the class $D_k^{ij}$ transmits the disease, it transmits the strain $i$ with probability $\mathbb{P}_k^{(i,j)\to i}$ and the strain $j$ with probability $\mathbb{P}_k^{(i,j)\to j}=1-\mathbb{P}_k^{(i,j)\to i}$. 
	
	Then we define the {\it probability that a  host in the class $k$ transmits $i$ if there is a contact}: 
	\begin{equation}\label{eq:defJki}
		J_k^i=I_k^i+\sum_{j=1}^N \left(\mathbb{P}_k^{(i,j)\to i}D_k^{ij}+\mathbb{P}_k^{(j,i)\to i} D_{k}^{ji}\right)
	\end{equation}
	
	Lastly, as in the one-strain situation, a host in any class $k$ may receive infection via contact with a host in class $s$ with probability $q_{ks}$. Thus, we define the probability for a host in class $k$ to be infected by strain $i$ as a summation over all classes:
	
	\begin{equation}\label{eq:defthetai}
		\Theta_k^i=\sum_{s\in\K} q_{ks} J_s^i
	\end{equation}
	Our system of co-infection, involving $N$ co-circulating strains, among different classes of hosts, reads with $X(0)\in\Omega_N$ and for $t>0$:

	\begin{equation}\label{mainsys}
		\begin{cases}
			\dfrac{d}{dt}{S}_k=  r_k(1 -S_k ) +\left(\sum\limits_{i=1}^N \gamma_k^i I_k^i+\sum\limits_{i=1}^N\sum\limits_{j=1}^N \gamma_k^{i,j} D_k^{i,j}\right) 
			-   S_k \sum\limits_{i=1}^N\beta_k^i  \Theta_k^i  \\
			\dfrac{d}{dt}{I}^i_k= \beta_k^i S_k \Theta_k^i - I^i_k \sum\limits_{j=1}^N \beta_k^j\sigma_k^{ij}\Theta_k^j
			-(r_k+\gamma_k^i)I^i_k,\quad k\in\K \\
			\dfrac{d}{dt}{D}_k^{ij}= \beta_k^j\sigma_k^{ij} \Theta_k^j I^i_k   -(r_k+\gamma_k^{i,j}) D_k^{ij}. \\
		\end{cases}\
	\end{equation}

	\begin{table}
		\centering
		\begin{tabular}{|c|c|c|}
			\hline
			Notation&Quasi-Neutral Formula&Meaning\\
			\hline
			\hline
			\multirow{2}{*}{$S_k(t)$}&\multirow{2}{*}{$S_k^*$}& Proportion of susceptible hosts \\
			&&in the class $k\in\K$ \\
			\hline
			\multirow{2}{*}{$I_k^i(t)$}&\multirow{2}{*}{$I_k^* z_i(\eps t)$}&Proportion of hosts in class $k\in\K$\\
			&& that are single-infected by strain $i\in \lrbN$\\
			
			\hline 
			\multirow{2}{*}{$D_k^{i,j}(t)$}&\multirow{2}{*}{$D_k^* z_i(\eps t)z_j(\eps t)$}& Proportion of hosts in class $k\in\K$\\&& that are double-infected by the strains $i$ then $j$.\\
			\hline
			\hline
			\multirow{2}{*}{$r_k$}&\multirow{2}{*}{$r_k$}& Per-capita birth and death rate \\
			&&of hosts in class $k$\\
			\hline
			\multirow{2}{*}{$\beta_k^i$}&\multirow{2}{*}{$\beta_k+\eps b_k^i$}& Transmission rate \\
			&&for the strain $i$ by individuals in class $k$\\
			\hline
			\multirow{2}{*}{$\gamma_k^i$}&\multirow{2}{*}{$\gamma_k+\eps c_k^i$}& Infection clearance rate of hosts in class $k$\\
			&& that are single-infected by strain $i$ \\
			\hline
			\multirow{2}{*}{$\gamma_k^{i,j}$}&\multirow{2}{*}{$\gamma_k+\eps c_k^{i,j} $}& Coinfection clearance rate for hosts in class $k$\\
			&& when co-infected by the strains $i$ then $j$\\
			\hline
			\multirow{2}{*}{$\sigma_k^{i,j}$}&\multirow{2}{*}{$\sigma_k+\eps \alpha_k^{i,j}$}& Ratio of double to single infection rate in class $k$\\
			&& for those double infected by the strains $i$ then $j$\\
			\hline

			\multirow{2}{*}{$\mathbb{P}_k^{(i,j)\to i}$}&\multirow{2}{*}{$\frac12+\eps w_k^{i,j}$}& Probability for a host in class $k$ \\&&double infected by the strains $i$ then $j$ to transmit $i$\\
			\hline
			
			\multirow{2}{*}{$\mathbb{P}_k^{(j,i)\to i}$}&\multirow{2}{*}{$\frac12-\eps w_k^{j,i}$}& Probability for a host in class $k$ \\&&double infected by the strains $j$ then $i$ to transmit $i$\\
			\hline
		\end{tabular}
		\caption{\textbf{Definition of the variables and the parameters in the model \eqref{mainsys}}. The second column gives the formula for the parameters and the state variables in the Quasi-Neutral regime. In this case, the state variables are expressed in terms of a product of the equilibrium $(S_k^*,I_k^*,D_k^*)_{k\in\K}$ of the no-strain system together with the frequency of each strain $(z_i)_{1\leq i\leq N}$. These frequencies follow the slow dynamics explicitly given by the replicator equation \eqref{main-replicator}. }
		\label{tabledef}
	\end{table}
	

		\subsection{$N-$ strain model: Neutrality}
			When the disease persists, the system \eqref{mainsys} is too complicated to allow a complete description of its dynamics. However, an important feature of the model is that when all the parameters do not depend on the strains, then the Neutral Null property applies, and we may describe completely its dynamics. 
		
		\begin{defi}[Neutral system]
			\label{def:neutral}
			The system \eqref{mainsys} is said to be Neutral if for any $k\in\K$ there exists three positive numbers $\beta_k$, $\gamma_k$ and $\sigma_k$ such that for all $i,j\in\lrbN$,
			\[\beta_k^i=\beta_k,\;\gamma_k=\gamma_k,\;\gamma_k^{i,j}=\gamma_k,\;
			\sigma_k^{i,j}=\sigma_k\] and if
			\[\mathbb{P}_k^{(i,j)\to i}=\mathbb{P}_k^{(j,i)\to i}=\dfrac12.\]
		\end{defi}
		Writing the model in terms of the aggregated variables 
		\begin{equation}\label{agregatedvar}
			I_k=\sum_{i\in\lrbN}I_k^i\text{ and }D_k=\sum_{(i,j)\in\lrbN^2} D_k^{ij}
		\end{equation} yields exactly to the no-strain model \eqref{SIDS} whose dynamics is completely described in the theorem \ref{th:SISqsk}. This yields the following proposition. 
		
		\begin{prop}[Aggregated variables]\label{Neutrall null 1 : aggregated variables.}
			Let $(S_k,\bI_k,\bD_k)_{k\in\K}$ be a solution of \eqref{mainsys}.
			Assume that the neutral assumption \ref{def:neutral} holds true.
			Define 		the aggregated variables 
            \[I_k=\bun_N^T \bI_k\text{ and }D_k=\bun_{N^2}^T \bD_k.\]
			
			If $\mathcal{R}_0\leq 1$ then  
			\begin{equation}
				\lim_{t\to +\infty}(S_k(t),I_k(t),D_k(t))_{k\in\K}=(1,0,0)_{k\in\K}=E_0
			\end{equation}
			If $\mathcal{R}_0>1$ then
			\begin{equation}
				\lim_{t\to +\infty}(S_k(t),I_k(t),D_k(t))_{k\in\K}=(S_k^*,I_k^*,D_k^*)_{k\in\K}=E^*
			\end{equation}
			where the steady state $E^*$  is defined in the theorem \ref{th:SISqsk}.
		\end{prop}
		
		\subsubsection{Rewriting the system in a convenient form}
        
		More can be said about the neutral model, and this is a key point in order relax the neutral assumption.
        
		In that respect, for any $k\in\K$ and $i\in\lrbN$, we define
		\[D_k^i=\sum_{j=1}^N\left(\mathbb{P}_k^{(i,j)\to i}D_k^{i,j}+\mathbb{P}_k^{(j,i)\to i}D_k^{j,i}\right)\]
		under the Neutral assumption we get simply
		\begin{equation}\label{Dki}
			D_k^i=\frac12 \sum_{j=1}^N \left(D_k^{i,j}+D_k^{j,i}\right).
		\end{equation}
		
		Then we find that $(I_k^i,D_k^i)$ satisfies an explicit non-autonomous linear system:
		\begin{equation}\label{IiDineutral}
			\dfrac{d}{dt}\begin{pmatrix}
				{I_k^i} \\
				{D_k^i}
			\end{pmatrix}_{k\in\K} = A(t) \begin{pmatrix} I_k^i\\ D_k^i \end{pmatrix}_{k\in\K}.
		\end{equation}
		
		This system has two important features.

First, the matrix $A(t)\in\mathcal{M}_{2p}$ is independent of the specific
strain $i$: all strains satisfy the same linear system.

Secondly, although the system is non-autonomous, the matrix $A(t)$ depends on
time only through the aggregated variables $S_k$, $I_k$ and $D_k$ ($k\in\K$),
and this dependence is polynomial. In particular, by
Proposition~\ref{Neutrall null 1 : aggregated variables.}, we obtain
$A(t)\to A^*$ exponentially fast as $t\to+\infty$.
		The explicit formula and the properties of the matrix $A^*$ play a central role in our approach and are outlined in the lemma \ref{lemmaAstar} below.

		\begin{lemma}[Neutral matrix]\label{lemmaAstar}
			Let the parameters be as in Table \eqref{tabledef}, with $\mathcal{R}_0>1$.  
			Denote by $E^*=(S_k^*,I_k^*,D_k^*)_{k\in\K}$ the endemic equilibrium of Theorem \ref{th:SISqsk} and set $T_k^*=I_k^*+D_k^*$ and $\Theta_k^*=(Q T^*)_k$.
			
			Define $A^*=\left(A_{ks}\right)_{(k,s)\in\K^2}\in\mathcal{M}_{2p}$ as the block matrix with $p\times p$ blocks $A_{ks}\in\mathcal{M}_2$ given by
			\[
			A_{ks}=q_{ks} B_k-\delta_k^s C_k,
			\]
			where $\delta_k^s$ is the Kronecker symbol,
			\[
			B_k=\beta_k \begin{pmatrix} S_k^*&S_k^*\\ \tfrac{\sigma_k}{2} I_k^*&\tfrac{\sigma_k}{2} I_k^*\end{pmatrix},\qquad
			C_k=(\gamma_k+r_k)\begin{pmatrix} 1&0\\0&1\end{pmatrix}+\sigma_k \beta_k \Theta_k^*\begin{pmatrix}1&0\\-\tfrac12&0 \end{pmatrix}.
			\]
			
			Using the Kronecker product, one has the concise representation of $A^*$:
			\[
			A^*=\diag(B_1,\ldots,B_p)\left(Q\otimes \mathbb{I}_2\right)-\text{diag}(r+\gamma)\otimes \mathbb{I}_2
			-\text{diag}(\sigma\beta\Theta^*)\otimes
			\begin{pmatrix}1&0\\-\tfrac12 &0 \end{pmatrix}.
			\]
			
			The matrix $A^*$ satisfies:
			\begin{enumerate}[label=(\roman*)]
				\item $A^*$ is Metzler (all off-diagonal entries are nonnegative).
				\item $A^*$ is irreducible.
				\item $A^*X^*=\bze_{2p}$, where $X^*=(I_k^*,D_k^*)_{k\in\K}$.
				\item $\alpha(A^*)=0$. i.e. $0$ is a simple eigenvalue of $A^*$, all other eigenvalues having negative real part.
				\item There exists a positive row vector $\bomega^*=(\phi_k^*,\psi_k^*)_{k\in\K}$ such that $\omega^* A^*=\bze_{2p}^T$ and $\omega^* X^*=1$.
			\end{enumerate}
			
			Explicitly,
			\[
			\forall k\in\K,\quad
			\phi_k^*= \frac{\q_k}{r_k+\gamma_k}\Bigl(1-\tfrac12\xi_k^*\Bigr)\mathcal{X}^*,\qquad
			\psi_k^*=\frac{\q_k}{r_k+\gamma_k}\mathcal{X}^*,
			\]
			with
			\[
			\xi_k^*=\frac{\sigma_k \beta_k\Theta_k^*}{r_k+\gamma_k+\sigma_k\beta_k \Theta_k^*}
			=\frac{D_k^*}{T_k^*}=\sigma_k\frac{I_k^*}{S_k^*},
			\]
			and normalization constant
			\[
			\mathcal{X}^*=\left(\sum_{k\in\K} \frac{\q_k}{r_k+\gamma_k}\Bigl(T_k^*-\tfrac12 \xi_k^* I_k^*\Bigr)\right)^{-1}.
			\]
			
			Where $\q>0$ is characterized in Corollary \ref{corxi} by $\q^T\text{diag}(\mathcal{R}S^* Q)=\q^T$ and $\sum_{k\in\K} \q_k=1$.  
			In particular, if $Q=q^T \bun$ then $\q=q$.
		\end{lemma}
		
		\begin{proof}
			See the section\ref{proofs:neutral33}.
		\end{proof}

\subsubsection{Asymptotic neutral dynamics}

			Using that $A(t)\to A^*$ and applying the lemma \ref{lemmaAstar} on the linear system $\dfrac{d}{dt}{X^i}=A^* X^i$ yields the main properties of the neutral system \eqref{IiDineutral}.
		For the statement of this lemma, we note the simplex of $\R^N$: 
		\[\Sigma^N=\{\bz=(z^i)_{1\leq i\leq N}\in [0,1]^N, \bun_N^T \bz=1\;\}\subset \R^N.\] 
		
		 \begin{lemma}[Asymptotic neutral dynamics]\label{lemma:Neutral}
Let $A^*\in\mathcal{M}_{2p}$ be as in Lemma~\ref{lemmaAstar}, and let 
$t\mapsto A(t)\in\mathcal{M}_{2p}$ satisfy $A(t)\to A^*$ exponentially fast as 
$t\to+\infty$.  

Consider the family $\mathbf{X}(t)=(X^i(t))_{i\in\lrbN}$  where for each $i\in\lrbN$ we have  $X^i(t)\in (0,1)^{2p}$  and
\begin{equation}\label{eq:At}
    \frac{d}{dt}X^i(t)=A(t)\,X^i(t).
\end{equation}
Let $X^*=(I_k^*,D_k^*)_{k\in\K}$ be as in Theorem~\ref{th:SISqsk}, and assume that \[
\sum_{i\in\lrbN} X^i(t)\longrightarrow X^*
\qquad\text{as }t\to+\infty.
\]
Let $\omega^*=(\phi_k^*,\psi_k^*)_{k\in\K}$ be as in Lemma~\ref{lemmaAstar}.  
For each $i\in\lrbN$, set $u^i(t)=\omega^*X^i(t)$ and 
$\bu(t)=(u^i(t))_{i\in\lrbN}$.

Then there exists $\bz\in\Sigma^N$ such that $\bu(t)\to\bz$ as $t\to+\infty$, and
\[
\mathbf{X}(t)\longrightarrow X^*\otimes\bz
\qquad\text{exponentially fast as }t\to+\infty.
\]
\end{lemma}

		\begin{proof}
        See the section \ref{proofs:neutral34}.
        \end{proof}

		Finally, we may rebuild the epidemiological quantities, obtaining the description of the complete epidemiological SIDS system under the perfect neutral assumption. 
		
\begin{thm}[Neutral dynamics]\label{Neutral:thm}
Assume that the neutral assumption \ref{def:neutral} holds.
Let $(S_k,\bI_k,\bD_k)_{k\in\K}$ be a solution of \eqref{mainsys}.

For any $i\in\lrbN$, let $X^i(t)=(I_k^i(t),D_k^i(t))^T_{k\in\K}$,
where $D_k^i=\dfrac12 \sum_{j=1}^N D_k^{ij}$.

For any strain $i\in\lrbN$, define $u^i(t)=\omega^* X^i(t)$ and
$\bu(t)=(u^i(t))_{i\in\lrbN}$,
where $X^{*}=(I_k^*,D_k^*)^T_{k\in\K}$ and
$\omega^*=(\phi_k^*,\psi_k^*)_{k\in\K}$ are given respectively in
Theorem~\ref{th:SISqsk} and Lemma~\ref{lemmaAstar}.

\begin{enumerate}[label=(\roman*)]

\item There exists $\bz=(z^i)_{i\in\lrbN}\in\Sigma^N$ such that
$\bu(t)\to \bz$. Hence $z^i$ represents the frequency of strain $i\in\lrbN$
within the host population.

\item As $t\to+\infty$, we have
\[
\bI(t)\to I^*\otimes \bz,\qquad
\bD(t)\to D^* \otimes\bz\otimes \bz,\qquad
\bTheta(t)\to \Theta^*\otimes \bz.
\]

In particular, for any $k\in\K$ and $i,j\in\lrbN$,
\[
\lim_{t\to +\infty} I^i_k(t)= z^i I_k^*
\quad\text{and}\quad
\lim_{t\to+\infty} D_k^{i,j}(t)= D^*_k\, z^i z^j.
\]

\item The set
\[
\mathcal{S}= \{(S_k^*,I_k^*\bz,D_k^*\bz\otimes \bz)_{k\in\K}\;:\; \bz\in \Sigma^N\}
\]
is invariant for \eqref{mainsys}. Equivalently,
\[
\mathcal{S}
=\{(\bS^*,\bI^*\otimes \bz, \bD^* \otimes \bz\otimes \bz)\;:\; \bz\in\Sigma^N\}.
\]

\end{enumerate}
\end{thm}

\begin{proof}
See Appendix~\ref{proofs:neutral35}.
        \end{proof}
		Firstly, the quantities $z^i$ may be seen as the frequency of the strains $i$ in the system. The attention of the reader is put on the fact that this definition is not trivial because the proportion of the strains $i$ depends on both the class $k$ and on the level of infection. 
		Thus, this theorem states explicitly how to choose the weights of each class in the computation of $z^i$ through the left eigenvector $\omega$ of $A^*$.
		
		Secondly, this theorem describes completely the structure of the Neutral system. There is an (asymptotically) invariant set $\Sigma^N\otimes \{X^*\}$ of dimension $N-1$ which   attracts all the trajectory. 
		
		Of course, this dynamics is not robust in the sense that a slight variation of the perfect Neutral assumption \eqref{def:neutral} will break this structure. However, if this variation from neutrality is small enough, most of the dynamics will be conserved. This is the object of the next section. 
		\subsection{$N-$ strain model: quasi-neutrality and the replicator equation}\label{sec:4}
			We relax the perfect identity assumption between strains \eqref{def:neutral} with the definition of Quasi-Neutrality.
		\begin{defi}[Quasi-Neutral system]
			\label{def:qneutral}
			Let $\eps\geq 0$. The system \eqref{mainsys} is said to be $\eps$-Neutral if  for any $k\in\K$ there exist three positive scalars $\beta_k$, $\gamma_k$ and $\sigma_k$  such that for any $i,j\in \lrbN$:
			\[|\beta^i_k-\beta_k|\leq \eps,\quad 
			|\gamma^i_k-\gamma_k|\leq \eps,\quad|\gamma^{i,j}_k-\gamma_k|\leq \eps;\quad
			|\sigma^{i,j}_k-\sigma_k|\leq \eps \text{ and }|\mathbb{P}_k^{(i,j)\to i}-\frac12 |\leq \eps.\]
			
			In practice, $\eps\ll 1$  and the system is referred to as Quasi-Neutral\footnote{Remark that a $0$-Neutral system is exactly a Neutral system.}.  The specific notation for each parameter is provided in Table \ref{tabledef}. 
		\end{defi}
		We now outline the principle of the slow–fast approximation. The key idea is that, under this definition, and by following the same steps as in the perfectly neutral system, the system \eqref{IiDineutral} takes the form
\[
\frac{d}{dt} X^{i,\varepsilon}
   = \left(A^* + \varepsilon U + \mathcal{O}(\varepsilon)\right) X^{i,\varepsilon}
     + \varepsilon\, \mathcal{F}^i\!\left(X^{1,\varepsilon},\ldots,X^{n,\varepsilon},\varepsilon\right),
\]
where \(\mathcal{F}^i\) is a nonlinear perturbation of the system, and the term \(\varepsilon U\) arises from the order‑one perturbation of the aggregated variables.

Following the idea of Lemma~\ref{lemma:Neutral}, we compute
\[
z^i = \lim_{\varepsilon\to 0} \omega^* X^{i,\varepsilon},
\qquad\text{where }\ \omega^* A^* = \mathbf{0}_{2p}.
\]
From Lemma~\ref{lemma:Neutral}, we obtain \(\bz = (z^i)_{i\in\lrbN} \in \Sigma^N\).

In contrast with the neutral case, for which \(\frac{d}{dt} z^i = 0\), we now obtain
\[
\frac{d}{dt} z^i
   = \varepsilon\left(
        \omega^* U X^{i,\varepsilon}
        + \omega^* \mathcal{F}^i(X^{1,\varepsilon},\ldots,X^{n,\varepsilon},0)
        + \mathcal{O}(\varepsilon)
     \right).
\]

To close this equation, we use Lemma~\ref{lemmaAstar}, which provides the perturbation
\[
X^{i,\varepsilon} = z^i X^* + \mathcal{O}(\varepsilon),
\]
and therefore
\[
\frac{d}{dt} z^i
   = \varepsilon\left(
        \omega^* z^i U X^*
        + \omega^* \mathcal{G}^i(\bz,\varepsilon)
        + \mathcal{O}(\varepsilon)
     \right),
\]
for some functions \(\mathcal{G}^i\).

Defining the slow time scale \(\tau = \varepsilon t\), dividing by \(\varepsilon\), and letting \(\varepsilon \to 0\) yields the slow equation for \(z^i\):
\[
\frac{d}{d\tau} z^i
   = z^i\left(\omega^* U X^* + \omega^* \mathcal{G}^i(\bz,0)\right).
\]

Using the key fact that \(\bz \in \Sigma^N\), we compute
\[
\frac{d}{d\tau}\sum_{i\in\lrbN} z^i
   = 0
   = \left(\omega^* U X^*\right)
     + \sum_{i\in\lrbN} z_i\, \omega^* \mathcal{G}^i(\bz,0),
\]
which implies
\[
\omega^* U X^*
   = - \sum_{i\in\lrbN} z_i\, \omega^* \mathcal{G}^i(\bz,0).
\]

Hence \(\bz\) satisfies the replicator equation
\[
\frac{d}{d\tau} z^i
   = z^i\,\omega^* \mathcal{G}^i(\bz,0)
     - \sum_{j\in\lrbN} z_j\, \omega^* \mathcal{G}^j(\bz,0),
\]
where only the quantities \(\omega^* \mathcal{G}^i(\bz,0)\) require explicit computation.

This argument leads to the main theorem of the paper.  
The formal proof relies on Tikhonov’s theorem and is postponed to Section~\ref{proofs:replicator}.
		\begin{thm}\label{Th:replicator}
			 
Assume that $\mathcal{R}_0>1$. 
We rely on the definitions of $S_k^*, I_k^*, D_k^*$ and $\Theta_k^*$ introduced in Theorem~\ref{th:SISqsk}, as well as on the definition of the stationary distribution 
$(\pi_k)_{k\in\K}$ given in Corollary~\ref{corxi} and the ratios $\xi_k^*=\frac{D_k^*}{T_k^*}$ in Lemma \ref{lemmaAstar}. 
            
			For any $\tau_0>0$ and $T>\tau_0$, there exists $\eps_0>0,$ $C>0$ and $\bz^0\in\Sigma^N$ such that for any $\eps\in(0,\eps_0)$,   
			if the system \eqref{mainsys} is $\eps$-quasi neutral then  the solution $(S_k(t),\bI_k(t),\bD_k(t))_{k\in\K}$ satisfies for any $\tau \in [\tau_0,T]$ : 
			
			\[\sum_{k\in\K}\left(\left\| S_k\left(\dfrac{\tau}{\eps}\right)-S_k^* \right\|+\sum_{i\in\lrbN}\left\| I_k^i\left(\dfrac{\tau}{\eps}\right)-I_k^* z^i (\tau)\right\|+\sum_{(i,j)\in \lrbN^2}\left\|D_k^{ij}\left(\dfrac{\tau}{\eps}\right)-D_k^* z^i(\tau) z^j(\tau)\right\|\right)\leq C \eps\]
			
			where the dynamics of strain frequencies $\bz$ are given by $\bz(0)=\bz^0$  and  for $\tau>0$ the replicator equation in $\Sigma^N$ : 
			
			\begin{equation}\label{main-replicator}
\frac{d}{d \tau} z^i=
z^i\biggl( \bigl( \Lambda \bz\bigr)_i - \bz^T \Lambda \bz \biggr),\quad i=1,\cdots,N\\
\end{equation}

			The payoff matrix is defined by $\Lambda = (\lambda_i^j)_{1\leq i,j\leq N}$ with  the { pairwise invasion  fitnesses between strains} are \begin{equation}\label{def:lambda}\lambda_i^j= \mathcal{X}^*\sum\limits_{k \in \mathcal{K}} \dfrac{\pi_k}{r_k+\gamma_k} Y^{i,j}_k\end{equation} where $\forall k\in \mathcal{K},$ 
			
			\begin{equation}\label{lambdaindetail}
				\begin{aligned} 
					Y_k^{i,j} =&
					\Theta_k^* S_k^*(b_k^i-b_k^j)
					-(1-\dfrac12 \xi_k^*) I_k^*(c_k^i-c_k^j)
					-D_k^*\left(\frac{c_k^{ij}+c_k^{ji}}{2}-c_k^{jj}\right)\\
					&+ \beta_k \sigma_k I_k^*\Theta_k^*  \left(w_k^{(i,j)}-w_k^{(j,i)}\right)\\
					&
                    +\frac12\beta_k\Theta_k^* I_k^*\xi_k^*\left(\left(\alpha_k^{ji}-\alpha_k^{jj}\right)+\mu_k\left(\alpha_k^{ji}-\alpha_k^{ij}\right)\right)
				\end{aligned}		
			\end{equation}		
			with $\mu_k=\dfrac{1-\xi_k^*}{\xi_k^*}=\dfrac{I_k^*}{D_k^*}=\dfrac{1}{\sigma_k \mathcal{R}_{k} \Theta_k^*}$. 
		\end{thm}
        \begin{proof}
    See section \ref{proofs:replicator}.
        \end{proof}
		\begin{rmk}
The last terms highlight the role of the key parameter 
\(\mu_k=\dfrac{I_k^*}{D_k^*}\), which has been discussed in detail in the 
single–host-class setting in \cite{gjini2020key}.

Let us denote \(T_k^0 = 1 - \dfrac{1}{\mathcal{R}_{k}}\) and 
\(\mu_k^0 = \dfrac{1}{\sigma_k \mathcal{R}_{k} T_k^0}\), the value of 
\(\mu_k\) in the single–class situation. We may then write
\[
\mu_k = \dfrac{T_k^0}{\Theta_k^*}\,\mu_k^0.
\]
The coefficient \(\dfrac{T_k^0}{\Theta_k^*}\) measures the strength of 
infection contributed by class \(k\) relative to the effective infection 
pressure in the multi–class setting.
\end{rmk}
	\subsubsection{Special case: Homogeneous $\cR_k$ Hosts}
		We consider the special case
\[
\frac{\beta_k}{r_k+\gamma_k}=\mathcal{R}_k=\mathfrak{R}_0 \qquad \text{for all } k\in\K,
\]
for a given common value $\mathfrak{R}_0$.

Without a strain structure, the endemic state is the positive solution of
\[
\dfrac{d}{dt}{T}_k = -(r_k+\gamma_k)T_k - \beta_k(1-T_k),
\qquad \Theta = QT.
\]

Since $\rho(Q)=1$, we obtain
\[
\mathcal{R}_0 = \rho\!\left(\diag(\mathcal{R}_k)\,Q\right)
= \mathfrak{R}_0\,\rho(Q)
= \mathfrak{R}_0.
\]
Thus, unsurprisingly, the threshold $\mathcal{R}_0$ coincides with the common value of the $\mathcal{R}_k$.

The disease persists if and only if $\mathcal{R}_0>1$.

By uniqueness and global stability of the endemic equilibrium, the limit is the same as in the one-class model. Namely,  for any $k\in\K$,
\[
\Theta_k(t)\longrightarrow \Theta^* := 1-\frac{1}{\mathcal{R}_0}
\qquad \text{as } t\to+\infty,
\]
and
\[
(S_k(t),I_k(t),D_k(t)) \longrightarrow (S^*,I^*,D^*)
\qquad \text{as } t\to+\infty,
\]
where
\[
S^*=\frac{1}{\mathcal{R}_0},\qquad T^*=1-\dfrac{1}{\cR_0},\qquad 
I^*=\frac{T^*}{1+\sigma\,\mathcal{R}_0\,\Theta^*},
\qquad
D^*=\sigma\,\mathcal{R}_0\,\Theta^*\,I^*.
\]

It follows that $\mathcal{X}^*=\dfrac{2\beta  T^*S^*}{2(T^*)^2-D^*I^*}.$

Now, introduced the strains and under the quasi neutrality assumption, Theorem \ref{Th:replicator} yields the following results\footnote{Note that homogeneity in $\mathcal{R}_k$ is not sufficient for this result. The mean parameters must be exactly the same within each host class.
}.
		 \begin{cor}\label{cor:hosts_hom}
            Assume that $\mathcal{\beta}_k$, $r_k$, $\gamma_k$ and $\sigma_k$ are independent from $k\in\K$. Assume that the system \eqref{mainsys-eps} is quasi-neutral.
            Denote for simplicity $\bar{x}=\mathbb{E}_\q(x_k)=\sum_{k\in\K}(\pi_k x_k)$
            
            Then the pairwize fitnesses' matrix $\Lambda=\left(\lambda_{i^j}\right)_{1\leq i,j\leq N}$ which drives the slow dynamics \eqref{main-replicator}  reduces to 
        
		\[\begin{aligned}\lambda_i^j=&
			\Xi_1 (\bar{b}^i-\bar{b}^j) + \Xi_2 
			(\bar{c}^j-\bar{c}^i) +\Xi_3 \left(\bar{c}^{jj}-\dfrac{\bar{c}^{ij}+\bar{c}^{ji}}{2}\right)
			\\&+\Xi_4(\bar{w}^{(i,j)}-\bar{w}^{(j,i)})+\Xi_5\left(\bar{\alpha}^{ji}-\bar{\alpha}^{jj}+\mu( \bar{\alpha}^{ji}-\bar{\alpha}^{ij})\right)
		\end{aligned}\]

where the weights $\Xi_i>0$, in front of each epidemiological trait contribution, are explicitly given by

\[\Xi_1=\dfrac{2TS^*(T^*)^2}{2(T^*)^2-D^*I^*} ;
	\quad\Xi_2=\dfrac{2T^*I^*-D^*I^*}{2(T^*)^2-D^*I^*} \;\]
		\[\Xi_3=\dfrac{
        2T^*D^*}{2(T^*)^2-I^*D^*};
		\quad\Xi_4=\dfrac{2\sigma \beta (T^*)^2 I^*}{2(T^*)^2-I^*D^*};\]

		\[  \Xi_5=\dfrac{\beta T^*I^*D^*}{2(T^*)^2-D^*I^*}\text{ and } 
		\mu=\dfrac{I^*}{D^*}\]

        \end{cor}
		All the mean parameters being class-independent, it follows that the strain invasion fitness parameters are simply given by the weighted mean of the perturbation from neutrality.
		
		In other words, it suffices to compute the mean variation from neutrality with respect to the stationary distribution $\q$ for every trait, and the system is \textit{mathematically} similar to the one-class system, but \textit{biologically}, the strain selection outcome may be different.

		\subsubsection{Special case: Host-independent perturbations}\label{Hostindependant}
		Here we focus on an important very natural special case where the strains variation is given independantly of the host structuration. This is the situation we are considering in all the three examples of the next section.

        If all the perturbations away from neutrality among strains do not depend on the host class ($k$) but only on the strains ($i,j$), then the replicator equation \eqref{main-replicator} may be rewritten in a much simpler form, which is very similar to the replicator equation for only one class of hosts (see \cite{le2023quasi}). 
		
        \begin{cor}\label{cor:strains_ind}
        Assume that the system \eqref{mainsys-eps} is quasi-neutral.
            Assume that all the deviation from neutrality given in table \ref{tabledef} are independant on $k\in\K$.
            Then the pairwize fitnesses' matrix $\Lambda=\left(\lambda_{i^j}\right)_{1\leq i,j\leq N}$ which drives the slow dynamics \eqref{main-replicator}  reduces to 
           
                \begin{equation}\label{eq:lambda_final}
		\lambda_i^j=\Xi_1 (b^i-b^j) + \Xi_2 (c^j-c^i) +\Xi_3 \left(c^{jj}-\dfrac{c^{ij}+c^{ji}}{2}\right)
		+\Xi_4(w^{(i,j)}-w^{(j,i)})+\Xi_5\left(\alpha^{ji}-\alpha^{jj}+\mu (\alpha^{ji}-\alpha^{ij})\right)
        \end{equation}
    
		where the weights $\Xi_i>0$, in front of each epidemiological trait contribution, are explicitly given by

\[\Xi_1=\mathcal{X}^*\sum_{k\in\K}\dfrac{\pi_k \Theta_k^*\mathcal{R}_{k}}{\beta_k} S_k^*;
		\quad\Xi_2=\mathcal{X}^*\sum_{k\in\K}\dfrac{\pi_k \mathcal{R}_{k}}{\beta_k}\left(1-\dfrac12 \dfrac{D_k^*}{T_k^*}\right) I_k^*;\]
		\[\Xi_3=\mathcal{X}^*\sum_{k\in\K} \dfrac{\pi_k \mathcal{R}_{k}}{\beta_k}D_k^*;
		\quad\Xi_4=\mathcal{X}^*\sum_{k\in\K}\pi_k \Theta_k^* \mathcal{R}_{k} \sigma_k I_k^*;\]

		\[ \text{and  } \Xi_5=\frac12\mathcal{X}^*\sum_{k\in\K} \pi_k \Theta_k^* \mathcal{R}_{k}\frac{I_k^*D_k^*}{T_k^*}\]  with 
		\begin{equation}\label{eq:mu_app}\mu=\sum_{k\in\K} h_k \mu_k,\text{ and } h_k=\frac{\pi_k\Theta_k^* \mathcal{R}_{k} \frac{I_k^*D_k^*}{T_k^*}}{\sum\limits_{s\in\K}\pi_s\Theta_s^* \mathcal{R}_{s} \frac{I_s^*D_s^*}{T_s^*}}.\end{equation}
	        
        \end{cor}
    	As can be seen above, the heterogeneity of the host population structure does not
affect the structure itself, but only the relative weight of importance of each
trait $\Xi_i$, as well as the skew-symmetry of the perturbation in
co-infection vulnerabilities through $\mu$.

For instance, in order to understand the effect of each trait’s variation on
the interaction between strains, one can directly apply the results of
\cite{le2023quasi} when strain dissimilarities are host-class independent.

In particular, when the within-strain variation affects only one trait, then for the
first four deviations from neutrality (with weights $\Xi_i$, $i=1,\dots,4$), host
heterogeneity does not change the qualitative nature of the dynamics, but only
their speed.
In contrast, as detailed in \cite{gjini2020key}, when the deviation from
neutrality arises through the co-infection coefficients $\alpha$, host
heterogeneity may affect the coefficient $\mu$, which can drastically modify the
\emph{quality} of the dynamics.

\subsubsection{Advantages of this replicator formulation: integrating pathogen and host variability}
This quasi-neutral slow–fast approximation and the resulting replicator equation
on the one hand drastically reduce the dimensionality of the system, and on the
other hand allow selection and competition dynamics between strains on arbitrary
host population structures to be studied more deeply and efficiently. It shows
explicitly—although this is far from trivial—how strain variability and host
population variability intertwine to govern the ultimate strain dynamics, their
quality, and their stability. Global strain dynamics in the host population can
be viewed as driven by a single global replicator that encapsulates all
micro-level information about strain traits, asymmetries, and epidemiological
variation in the population. This can greatly simplify our understanding of
multi-strain propagation phenomena in realistically complex populations. Next, we
illustrate the utility of this framework through several applications and the
analysis of special cases, focusing on host-independent perturbations in the
co-infection coefficients, which lead to studying the dependence of the scalar
$\mu$ on host heterogeneity.
		\section{Applications}\label{sec:application}
		\subsection{Application 1: Two classes of hosts}
		One possible application of this model is when the population of the host is divided between two classes as two different host species, the effect of a treatment on the host... Denote $\K=\{A,B\}$, the basic no-strain model reads explicitly 
		\begin{equation}\label{SIDv}
			\begin{cases}
				\dfrac{d}{dt}{S}_k=(r_k+\gamma_k)(1 - S_k)-\beta_k S_k \Theta\\
				\dfrac{d}{dt}{I}_k= \beta_k S_k \Theta -\sigma_k \beta_k 
				I_k \Theta  -(\gamma_k+r_k) I_k \\
				\dfrac{d}{dt}{D}_k= \sigma_k\beta_k 
				I_k \Theta  -(r_k+\gamma_k) D_k \\
			\end{cases}
		\end{equation}
		with $k\in\{A,B\}$ and $\Theta=q_A (1-S_A)+q_{B}(1-S_{B})$.

		The model is completely described by the equations on $S_A$ and $S_B$ only. 
		\begin{equation}\label{SSV}
			\begin{cases}
				\dfrac{d}{dt}{S}_A=(r_A+\gamma_A)(1 - S_A)-\beta_A S_A (1-q_AS_A-q_{B}S_{B})\\
				\dfrac{d}{dt}{S}_{B}=(r_{B}+\gamma_{B})(1 - S_{B})-\beta_{B} S_{B} (1-q_AS_A-q_{B}S_{B})\\
			\end{cases}
		\end{equation}
		This model is exactly the SIS model with two classes of hosts, which is well-known among cases of population heterogeneity (see \cite{diekmann2010construction}).
		We get explicitly that the disease is endemic if and only if $q_A\mathcal{R}_A+q_{B}\mathcal{R}_{B}>1,$
        where we have $\mathcal{R}_k=\dfrac{\beta_k}{r_k+\gamma_k}.$
        In that case $(S_k,I_k,D_k)\to (S_k^*,I_k^*,D_k^*)$ for $k=A,B$ where 
		 
		\[S_k^*=(1+\mathcal{R}_k \Theta^*)^{-1},\quad
		I_k^*=\dfrac{\mathcal{R}_k \Theta^*}{1+\sigma_k \mathcal{R}_k \Theta^* S_k^*},\quad 
		D_k^*=\sigma_k \mathcal{R}_k \Theta^* I_k^*\]
		and
		\[\Theta^*=\dfrac{1}{2}\left(1-\dfrac{1}{\mathcal{R}_A}-\dfrac{1}{\mathcal{R}_B}+\sqrt{q_A \left(1-\dfrac{1}{\mathcal{R}_A}+\dfrac{1}{\mathcal{R}_B}\right)^2
			+q_B \left(1-\dfrac{1}{\mathcal{R}_B}+\dfrac{1}{\mathcal{R}_A}\right)^2
		}\right).\]
		
		The replicator equation that follows is then fully explicit. In particular, it shows that changing only $R_A$ and $R_B$ may impact drastically the output even in the simple 2-strain situation.  
		
		In the case of 2 strains, the dynamics are completely driven by the signs of $\lambda_1^2$ and $\lambda_2^1$. There are 4 possibilities: 
		\begin{itemize}
			\item 
			Strain 1 wins in competitive exclusion $(\lambda_1^2>0,\lambda_2^1<0)$,
			\item Strain 2 wins in competitive exclusion $(\lambda_1^2<0,\lambda_2^1>0)$, 
			\item Both strains coexist  $(\lambda_1^2>0,\lambda_2^1>0)$ 
			\item Bi-stability (only one strain wins depending on the initial condition)  $(\lambda_1^2<0,\lambda_2^1<0)$. 
		\end{itemize}		
		As described in \cite{gjini2020key}, in a single host class situation, it is shown that for a fixed value of the perturbation from neutrality,  a variation in the global epidemiological parameters, e.g. transmission intensity $\mathcal{R}_0$ or mean susceptibility to coinfection  $\sigma$, may drastically change the epidemiological competition between strains. The same phenomenon occurs obviously here, but with the additional effect of the presence of two host classes that may change the outcome. 
		
		\paragraph{\textbf{Strain variation in co-infection susceptibility factors.}} For simplicity, focus on the case where there is a perturbation only in $\alpha_{ij}$ (e.g. \cite{madec2020predicting}). The dynamics are then driven by the signs of the pairwise invasion fitnesses:
		
		\[\lambda_1^2=\alpha^{21}-\alpha^{22}+\mu (\alpha^{21}-\alpha^{12})\]
		\[\lambda_2^1=\alpha^{12}-\alpha^{11}+\mu (\alpha^{12}-\alpha^{21})\]
		
		As  functions of $\mu$, $\lambda_1^2$ and $\lambda_2^1$ change sign respectively at the explicit critical values:
		
		\[\mu_{crit}^1=\dfrac{\alpha^{21}-\alpha^{22}}{\alpha^{21}-\alpha^{22}},\quad \mu_{crit}^2=\dfrac{\alpha^{11}-\alpha^{12}}{\alpha^{21}-\alpha^{22}}.\]
		
		Due to the variation in all the global parameters, $\mu$ may cross these values, resulting in a change of the quality of the dynamics. 
		To illustrate that, let us fix all the parameters but $q_A$ (a proxy for varying host population structure).

		Notice $\mu=\mu(q_A)$ is a non-linear function of $q_A$.
		Of course we have $\mu(0)=\mu_B=(\sigma_B(\mathcal{R}_B-1))^{-1}$ and $\mu(1)=\mu_A=(\sigma_A(\mathcal{R}_A-1))^{-1}$. In between, 
		$\mu(q_A)$ depends explicitly but neither linearly nor monotonously on  $q_A$. This non-linearity allows us to construct interesting phenomena. For instance, it is possible that $\mu(0)=\mu(1)$ and then $\lambda_i^j(0)=\lambda_i^j(1)$ but for intermediate values of $q$, $\mu$ may cross the critical values $\mu_{crit}^1$ and $\mu_{crit}^2$ resulting in a change in the signs of the $\lambda$'s and thus on the qualitative behavior of the 2-strain dynamics.  
		
		Such an example is given in Figure \ref{fig:exampletwohost}.  In this example, 
			we have chosen the parameters as follows. For the strains, we have set $\alpha^{11}>\alpha^{12}>\alpha^{21}>\alpha^{22}$. Ecologically speaking, this means that strain 1 is a good co-colonizer with itself, but is easy to be co-colonized by the other strain, and strain 2 is a bad co-colonizer with itself but is resistant to co-colonization by the other. Mathematically, this implies that $\lambda_1^2$ and $\lambda_2^1,$ respectively, decrease from positive to negative values, and increase from negative to positive values, as $\mu$ increases from $0$ to $+\infty$.
			
			For the host class parameters, we have set
			$\mathcal{R}_A<\mathcal{R}_B$ and $\sigma_A>\sigma_B$ such that $\sigma_A (\mathcal{R}_A-1)=\sigma_B (\mathcal{R}_B-1)$. The class $A$ is then a class with a small prevalence of infection but a lot of co-infection, and the class $B$ is the opposite. However, from the point of view of epidemiological interaction between strains, which is driven by $\mu_A=\mu_B$, each of the two classes alone shows the same qualitative strain dynamics.  
			This setup leads to a non-monotonous dependence of $\mu$ on $q_A$ and thus of $\lambda_1^2$ and $\lambda_2^1$ as well.

            In Figure \ref{fig:exampletwohost}, we observe that strain 2 excludes strain 1 in both classes of host when isolated ($q_A=0$ or $q_A=1$). However, thanks to heterogeneous population structure, in the case of epidemiological transmission between these classes (intermediate value of $q_A$), strain 1 may: be able to survive, coexist with strain 2, or even to exclude strain 2 from the system, depending on particular parameter values. This phenomenon results from our specific choice of both the host-dependent epidemiological parameters $\mathcal{R}_k$ and $\sigma_k$ and the strain-dependent deviation from neutrality $\alpha$'s, but is broadly illustrative of typical qualitative shifts in strain dynamics induced by population-level heterogeneity.
			
		\begin{figure}[h!]
			\centering
			\includegraphics[width=0.7\linewidth]{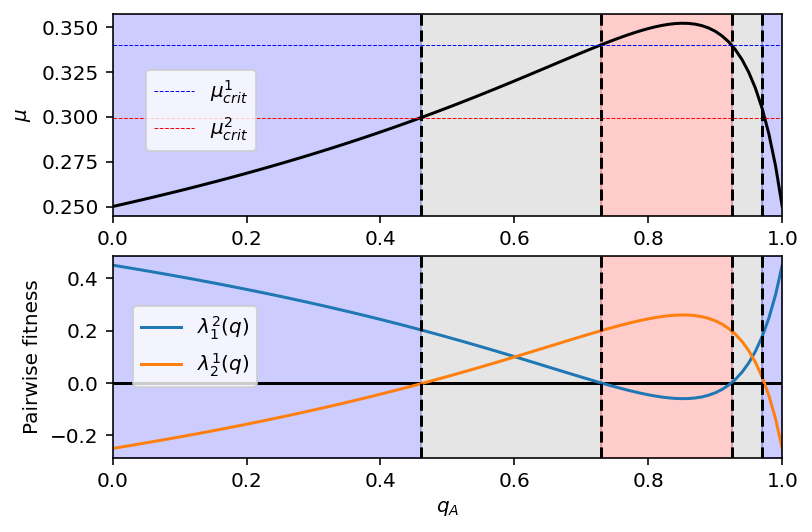}
			\caption{\textbf{Model illustration for 2-strain SIS system outcomes on a two-class host population, each with a different basic-reproduction number and mean coinfection vulnerability.} The strain 1 is the only survivor in both classes $A$ and $B$ ($\lambda_1^2>0,\lambda_2^1<0$ in both extremes $q_A=0,q_A=1$). By continuity, this is true when $q_A\approx 0$ or $q_A\approx 1$ (filled in blue). For intermediate values of $q_A$, there are 3 possibilities (bottom panel): i) strain 1 still is the only survivor (blue-shaded region), ii) both strains coexist (filled in gray) or iii) strain 2 is the only survivor (filled in red). The shifts between these qualitative regimes occur exactly at the values when $\mu$ of the heterogeneous population intersects with $\mu_{crit}^1$ and $\mu_{crit}^2$ (top panel).
				The epidemiological parameters are $\mathcal{R}_A=1.5$, $\mathcal{R}_B=5$, $\sigma_A=8$ and $\sigma_B=1$ so that $\mu_A=\mu_B=\frac14$. The perturbations from neutrality among strains are in pairwise vulnerabilities to coinfection \cite{madec2020predicting}: $(\alpha_{ij})=\begin{pmatrix}5.1&4\\-1&-2.7 \end{pmatrix}.$ }
			\label{fig:exampletwohost}
		\end{figure}

		\subsection{Application 2: Structure induced by vaccination}

The previous example may be applied in an intervention scenario such as a
vaccination model with two classes: A, vaccinated hosts, and B, non-vaccinated
hosts.

Here, however, we can explore this framework in slightly more detail by allowing
a heterogeneous distribution of protection effects among vaccinated hosts, in
the context of a universal vaccine (equally effective against all strains).
\subsubsection{Model formulation}
\paragraph{Vaccination effect.}
Within this model, we can implement the effect of an intervention or treatment
through four distributions in the host population: $\gamma_k$, the clearance
rate (for instance, under antibiotic treatment); $\beta_k$, the infection rate
of class $k$; $q_k$, the transmission rate of class $k$; and $\sigma_k$, the
ratio between primary and secondary infection rates. For simplicity, we assume
that all hosts share the same clearance rate $\gamma$ and the same coinfection
vulnerability factor $\sigma$. Thus, we examine homogeneous and heterogeneous
vaccination effects through changes in per-capita infection rates.

\noindent\textbf{No intervention population: homogeneous baseline $\beta_0$.}
The model is given by the following equations\footnote{Here the index $k$ serves
only to track the later emergent discrete population structure.}:
\begin{equation}\label{SID-novaccin}
\begin{cases}
\dfrac{d}{dt}{S}_k=r(1 - S_k)-\beta_0 S_k \Theta +\gamma (I_k +D_k ),\\[0.3em]
\dfrac{d}{dt}{I}_k= \beta_0 S_k \Theta -\sigma \beta_0 I_k \Theta  -(\gamma+r) I_k,\\[0.3em]
\dfrac{d}{dt}{D}_k= \sigma\beta_0 I_k \Theta  -(r+\gamma) D_k,
\end{cases}
\qquad
\Theta=\sum_{k\in\K} q_k (I_k+D_k).
\end{equation}

\noindent\textbf{Homogeneous intervention: $\beta=(1-v)\beta_0$.}
We assume full vaccination coverage. The model has the same structure as above,
but all hosts now experience a uniformly reduced infection rate $\beta$:
\begin{equation}\label{SID-vaccinhom}
\begin{cases}
\dfrac{d}{dt}{S}_k=r(1 - S_k)-\beta S_k \Theta +\gamma (I_k +D_k ),\\[0.3em]
\dfrac{d}{dt}{I}_k= \beta S_k \Theta -\sigma \beta I_k \Theta  -(\gamma+r) I_k,\\[0.3em]
\dfrac{d}{dt}{D}_k= \sigma\beta I_k \Theta  -(r+\gamma) D_k,
\end{cases}
\qquad
\Theta=\sum_{k\in\K} q_k (I_k+D_k).
\end{equation}

\noindent\textbf{Heterogeneous intervention: $\beta_k=(1-v_k)\beta_0$.}
Here, vaccinated hosts may experience different infection rates after
vaccination. Vaccination alone therefore generates heterogeneity in
transmission, manifested as a discrete structure in the $\beta_k$ values.\footnote{The
case of heterogeneous intervention combined with pre-existing heterogeneity will
be addressed in future work.} The system becomes a genuine multi-host-class
model:
\begin{equation}\label{SID-vaccink}
\begin{cases}
\dfrac{d}{dt}{S}_k=r(1 - S_k)-\beta_k S_k \Theta +\gamma (I_k +D_k ),\\[0.3em]
\dfrac{d}{dt}{I}_k= \beta_k S_k \Theta -\sigma \beta_k I_k \Theta  -(\gamma+r) I_k,\\[0.3em]
\dfrac{d}{dt}{D}_k= \sigma\beta_k I_k \Theta  -(r+\gamma) D_k,
\end{cases}
\qquad
\Theta=\sum_{k\in\K} q_k (I_k+D_k).
\end{equation}

Different distributions of vaccine effects may yield the same overall basic
reproduction number $\mathcal{R}_0$ post-vaccination, while differing in the
details of protection across individuals.

We assume that strain deviations from neutrality are independent of the host
class. In the homogeneous vaccination case, the intervention primarily induces a
change from $\mathcal{R}_0^{ori}=\frac{\beta_0}{\gamma+r}$ to
$\mathcal{R}_0^{hom}=\frac{\beta_0(1-v)}{\gamma+r}$, which subsequently affects
the total prevalence of infection and co-infection, and may also influence
strain dynamics through downstream effects on $\Lambda$ (see also
\cite{le2022disentangling}).

In the heterogeneous vaccination case, the variation $\beta_k=\beta_0(1-v_k)$
generates a distribution of $\mathcal{R}_k$. The basic reproduction number is
\begin{equation}
\mathcal{R}_0=\mathbb{E}_q(\mathcal{R}_k)=\dfrac{\beta_0}{r+\gamma}\bigl(1-\mathbb{E}_q(v_k)\bigr).
\end{equation}

Hence, for the same mean efficacy $\mathbb{E}_q(v_k)$, the basic reproduction
number remains the same. From the point of view of disease elimination (i.e.,
the condition $\mathcal{R}_0\leq 1$), heterogeneity plays no role.

However, if $\mathcal{R}_0>1$, heterogeneity does matter. First, according to
Corollary~\ref{prop:SIDS}, the total prevalence $\Theta^*$ is smaller in the
heterogeneous vaccination case. This indicates that a vaccine may be overall more
effective if some host classes respond better than others.

\paragraph{Strain interaction.}
When studying strain interactions, the impact of heterogeneous vaccination is
even more pronounced. Heterogeneity in the $v_k$ can further interfere with
strain selection, even when the overall prevalence is the same as under a
homogeneous vaccine.

In the specific case where strain variability is independent of host responses
to vaccination, Corollary~\ref{cor:strains_ind} yields the following pairwise
invasion coefficients:
\[
\begin{aligned}
\lambda_i^j
=&\;
\Xi_1 (b^i-b^j)
+ \Xi_2 (c^j-c^i)
+ \Xi_3 \!\left(c^{jj}-\frac{c^{ij}+c^{ji}}{2}\right)
\\
&\;
+ \Xi_4\bigl(w^{(i,j)}-w^{(j,i)}\bigr)
+ \Xi_5\!\left(\alpha^{ji}-\alpha^{jj}+\mu(\alpha^{ji}-\alpha^{ij})\right),
\end{aligned}
\]
where the effect of vaccination appears in the $\Xi_i$, in $\mu$
(see also Section~\ref{Hostindependant}).

\subsubsection{Exploration of the role of vaccine heterogeneity on  multi-strain selection}
From Proposition~\ref{prop:SIDS}, we already know that for a fixed mean vaccine
effect $\mathbb{E}_q(v_k)$, and thus a fixed basic reproduction number
$\mathcal{R}_0$, a heterogeneous distribution yields a greater reduction in
total prevalence $\Theta^*$.

We now ask how such heterogeneity affects strain
interaction and, consequently, epidemiological selection. As before, we focus
here on the effect of host heterogeneity on $\mu$ (see
Figure~\ref{fig:vaccinmu}).

\begin{figure}[h!]
\centering
\includegraphics[width=0.9\linewidth]{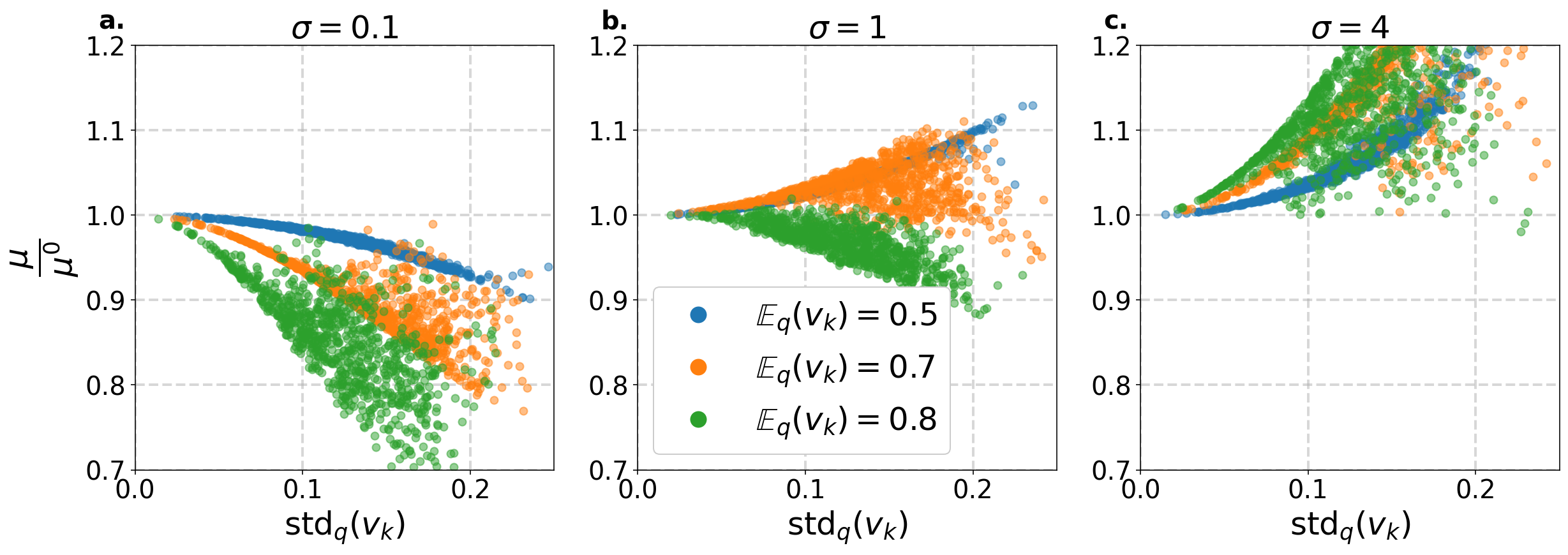}
\caption{\textbf{Vaccination-induced heterogeneity in
$\beta_k=\beta_0(1-v_k)$ may affect strain selection through its effect on
$\mu$ given in \eqref{eq:mu_app}.} The heterogeneity of $\mathcal{R}_k=\beta_k(r+\gamma)^{-1}$ may result
in a smaller or larger $\mu$ than the homogeneous value $\mu^0$. The relative
effect depends on the coinfection vulnerability factor $\sigma$, on the mean
vaccine efficacy relative to $q$ ($v=\mathbb{E}_q(v_k)$), and on the
heterogeneity of the vaccine-effect distribution
$\mathrm{std}_q(v_k)=\mathbb{E}_q(v_k^2)-v^2$. On average, small values of
$\sigma$ and large values of $v$ lead to smaller $\mu$ under vaccination, while
large $\sigma$ and small $v$ lead to larger $\mu$. The effect strengthens as
$\mathrm{std}_q(v_k)$ increases. Hence, transmission heterogeneity induced by
vaccination may alter $\mu$ and thereby affect the dynamics and final outcomes
of interacting pathogens. Here $\beta_0=10$.}
\label{fig:vaccinmu}
\end{figure}

We find that, everything else equal, the heterogeneity (measuring by $\mathrm{std}_q(v_k)$) of the vaccine monotonically change the value of $\mu$.
But $\mu$ may be increasing or decreasing with the heterogeneity depending on the global epidemiological parameters as $\sigma$ and $\beta_0$ and on the mean vaccine efficiency  $\mathbb{E}_q(v_k)$.
Ours simulations indicate that for small $\sigma$, sand small $\beta_0$ and for large mean vaccine efficiency, $\mu$ decrease with the heterogeneity of the vaccine. In contrast,
for large $\sigma$,  $\beta_0$ and for small mean vaccine efficiency, $\mu$ increase with the heterogeneity of the vaccine.

The global effect of $\mu$ has been studied independently of any particular host structure in \cite{gjini2020key}. 
This work shows that small $\mu$ tends to stabilize the dynamics while creating multiple stationary attractors with only a few persisting strains, 
whereas large $\mu$ tends to destabilize the system and favors more complex behavior such as cycles or chaos, allowing more strains to persist.

Hence, applying to this vaccine setup,  our simulations indicates that 
\begin{enumerate}[label=(\roman*)]
\item When the global mean parameters of the disease are low after vaccination (i.e. small $\beta_0$, small $\sigma$ or large $\mathbb{E}_q(v_k)$) the heterogeneity of the vaccine stabilize the dynamics but favorises multiple attractors. 
\item When the global mean parameters of the disease are large after vaccination (i.e. large $\beta_0$, large $\sigma$ or small $\mathbb{E}_q(v_k)$) the heterogeneity of the vaccine destabilize the dynamics yielding to complex attractors. 
\end{enumerate}
This situation is illustrated in figure \ref{fig:vaccinmudyn} wherein we have set the deviation from neutrality
\begin{equation}\label{alphavaccin}(\alpha_{ij})_{1\leq i,j\leq 4}=\begin{pmatrix}
    -0.72& -0.25 & -0.44& -0.38\\
     -0.07&  0.66&  0.6& -0.51\\
      0.99& -0.26&  0.4&  0.43\\
      0.92& -0.61& -0.26&  0.5
\end{pmatrix}.\end{equation}

\begin{figure}[h!]
\centering
\includegraphics[width=0.9\linewidth]{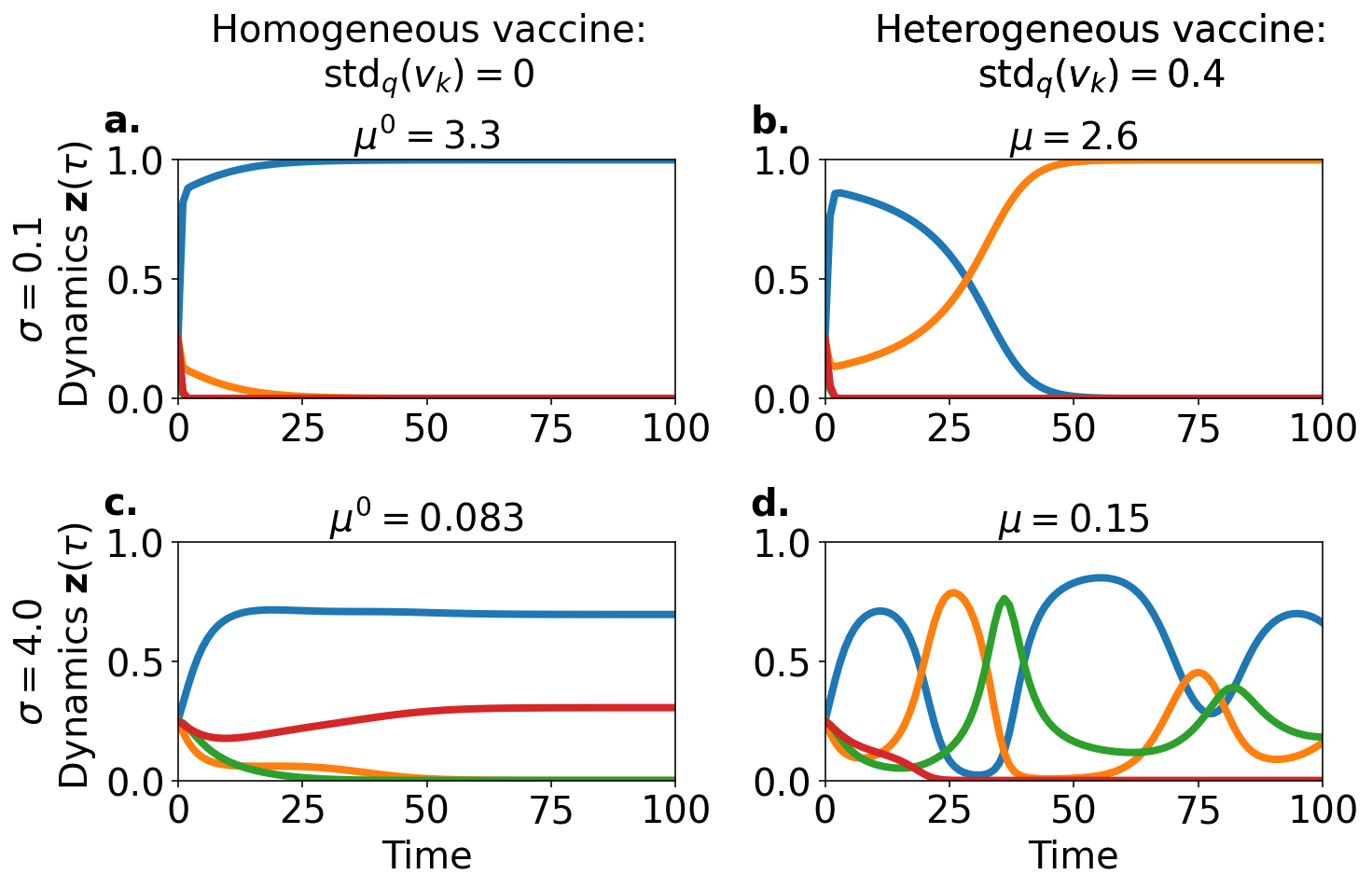}
\caption{
\textbf{Vaccination-induced heterogeneity alters strain selection through its effect on $\mu$ given in \eqref{eq:mu_app}.}
We illustrate how $\mu$ pre- and post-vaccine shapes four-strain replicator dynamics. All simulations use $\beta_0=10$, deviation from neutrality $\alpha$ defined in \eqref{alphavaccin} and mean vaccine efficacy of 50\%: $\mathbb{E}_q(v_k)=0.5$. 
The first row corresponds to a multi-strain system with relatively more competition among strains $\sigma=0.1$ (lower coinfection vulnerability factor $<1$), and the second row to a system with relatively more facilitation among strains $\sigma=4$ (high coinfection vulnerability factor $>1$). 
Left panels show dynamics under a homogeneous universal vaccine; right panels show dynamics under a universal vaccine with the same mean efficacy but variable protection across individuals $\mathrm{std}_q(v_k)=0.4$. 
\textbf{a-b.} Heterogeneity in vaccine effect decreases $\mu$ for small $\sigma$ (low coinfection propensity), and changes drastically the outcome of strain selection. \textbf{c-d.} The same heterogeneity in vaccine effect can markedly increase $\mu$ for large $\sigma$ (higher coinfection propensity), again significantly modifying strain selection.
}
 
\label{fig:vaccinmudyn}
\end{figure}

        \subsection{Application 3: Heterogeneous host contact network}
		Following \cite{moreno2002epidemic}, this model may be applied to a context of a complex host contact network, a case we study in detail in a companion paper \cite{networkPaper}.
		\subsubsection{Model formulation}
		Each host may interact with others though a static network of connections. 
		
		Here, the classes $\K$ describe the possible number of contacts of hosts in the population. For our purpose, the network $\mathcal{N}=\{\K,p\}$ is completely characterized by the proportion of nodes of degree $k\in\K$ though the probability law $p=(p_k)_{k\in\K}$. The mean connectivity of a node is then 
		\[\mathbb{E}_p(k)=\sum_{k\in\K} k p_k.\]
		 
		Using the mean-field approximation approach \cite{moreno2002epidemic}, the epidemic model is constructed as follows.
		
		The infective rate of a class is proportional to the number of contacts, times the probability of infection per contact $\lambda$:
		\[\beta_k=\lambda k.\]
		The probability of being infected by a class $k$ is proportional to the relative frequency of infected hosts in this class: $T_k=I_k+D_k=(1-S_k)$ and to the probability of being in contact with a class $k$ is:
		\[q_k=\dfrac{k p_k}{\E_p(k)}.\]
		
		The probability to be infected is proportionate to the mean field parameter \[\Theta=\dfrac{\E_p(k T_k)}{\E_p(k)}=\E_q(T_k).\]
		
		Since hosts differ only by their degree, all other parameters are independent of $k$.
		We obtain the following network system:
		
		\begin{equation}\label{SID-betaconstant}
			\begin{cases}
				\dfrac{d}{dt}{S}_k=r(1 - S_k)-\lambda k S_k \Theta +\gamma (I_k +D_k )\\
				\dfrac{d}{dt}{I}_k= \lambda k S_k \Theta -\sigma \lambda k 
				I_k \Theta  -(\gamma+r) I_k \\
				\dfrac{d}{dt}{D}_k= \sigma\lambda k 
				I_k \Theta  -(r+\gamma) D_k \\
			\end{cases},\;
		\end{equation}
		We set $\rho=\dfrac{\lambda}{r+\gamma}$ so that $\mathcal{R}_k=\rho k$.  According to theorem \ref{th:SISqsk},  the disease persists if and only if\footnote{Here, we use the notation $\mathcal{R}_0=\mathcal{R}_{\mathcal{N}}$ in order to insist on the effect of the network $\mathcal{N}$.} \[\mathcal{R}_{\mathcal{N}}:=\rho \dfrac{\E_p(k^2)}{\E_p(k)}>1.\] In that case $\Theta(t)\to \Theta^*\in(0,1)$ the only solution of 
		\[\rho\sum_{k\in\K} \dfrac{k^2p_k}{1+\rho k \Theta}=\E_p(k)\]
		
		The global attractor is defined by the following:
		\[S_k^*=(1+\rho k\Theta^*)^{-1},\; T_k^*=\rho\Theta^* k S_k^*,\quad I_k^*= \dfrac{\rho k \Theta^*}{1+\sigma\rho k\Theta^*} S_k^*\text{ and } D_k^*=\sigma\Theta^* \rho k I_k^*.\]

        Next, we focus on the interaction between strains.
		Since the population structuring depends only on host behavior/traits: the number of connections, it is natural to assume that the strain traits are independent of the hosts' classes. Then the replicator system is fully described by the quantities:

		\[
		\lambda_i^j=\Xi_1 (b^i-b^j) + \Xi_2 (c^j-c^i) +\Xi_3 \left(c^{jj}-\dfrac{c^{ij}+c^{ji}}{2}\right)
		+\Xi_4(w^{(i,j)}-w^{(j,i)})+\Xi_5\left(\alpha^{ji}-\alpha^{jj}+\mu (\alpha^{ji}-\alpha^{ij})\right)\]
		where the weights $\Xi_i>0$ are explicitly given by
		
		\[\Xi_1=\dfrac{\Theta^*\mathcal{X}^*}{(r+\gamma) \bar{k}}\sum_{k\in\K}p_kk S_k^*;
		\quad\Xi_2=\dfrac{\rho \Theta^*\mathcal{X}^*}{(r+\gamma) \bar{k}}\sum_{k\in\K}p_k k^3 \left(1-\dfrac12 \dfrac{D_k^*}{T_k^*}\right) \dfrac{S_k^*}{1+\sigma\rho k \Theta^*};\]
		\[\Xi_3=\dfrac{\sigma\mathcal{X}^*(\Theta^*\rho)^2}{\bar{k}(r+\gamma)}\sum_{k\in\K}p_kk^3\dfrac{ S_k}{1+\sigma \rho k \Theta^*};
		\quad\Xi_4=(r+\gamma)\Xi_3,\quad \Xi_5=\dfrac{(r+\gamma)}{2\sigma} \Xi_3,\] and
		\begin{equation}\label{munet}		\mu=\dfrac{1}{\sigma \rho \Theta^*}\sum_{k\in\K} h_k \dfrac{1}{k},\text{ with } h_k=p_k\frac{k^4}{(1+\sigma k\rho\Theta^*)^2(1+\rho k \Theta^*)}\left(\sum\limits_{s\in\K}p_s\frac{s^4}{(1+\sigma s\rho\Theta^*)^2(1+\rho s \Theta^*)} \right)^{-1}.\end{equation}
		
        \subsubsection{Analyzing the role of population structure heterogeneity}
        \label{sec:five}
        This quasi-neutral and replicator framework for multi-strain dynamics (Theorem \ref{Th:replicator}) allows to analyze more easily and transparently the precise role of heterogeneity in host population structure on general endemic prevalence of multi-strain pathogens and on their selective dynamics over time. We do not intend to be exhaustive in this exploration, but we address this issue in the simple case of host population structure manifested in the host contact degree distribution. We will analyze starting from the multi-strain system presented in the last example, corresponding to dynamics on a host contact network. 
    
		We compare an arbitrary host contact network $\mathcal{N}$ with the corresponding homogeneous network $\mathcal{N}_0$ with only one class with connectivity\footnote{If $\E_p(k)\notin\mathbb{N}$ then this is just a one-class SIDS system with $\beta=\rho \E_p(k)$. We use the vocabulary of a homogeneous network $\mathcal{N}_0$ to highlight the effect of the heterogeneity of the network $\mathcal{N}$.\label{footnote}} $\E_p(k)$. First, we compare the global endemic persistence quantities in the two scenarios (Lemma \ref{lemma:net0}), then we explore the strain selection (Lemma \ref{mubar}).
		\begin{lemma}\label{lemma:net0}
			Let $\mathcal{N}=(\K,(p_k)_{k\in\K})$ be a network with mean connectivity $k_0=\E_p(k)$. Denote $\mathcal{R}_\mathcal{N}=\rho \dfrac{\E_p(k^2)}{\E_p(k)}$.
			
			Let  $\mathcal{N}_0$ be the homogeneous network with constant connectivity $\E_p(k)$ and denote $\mathcal{R}_{\mathcal{N}_0}=\rho \E_p(k)$
			
			Assume that $\mathcal{N}\neq \mathcal{N}_0$.
			\begin{enumerate}[label=(\roman*)]
				
				\item In $\mathcal{N}$, we have $\Theta(t)=\mathbb{E}_q(T_k (t))$ converges to  $\Theta_{\mathcal{N}}:=\Theta^*\in[0,1]$.
				
				The disease persists, that is $\Theta^*>0$, if and only if $\mathcal{R}_{\mathcal{N}}>1$.
				
				\item In $\mathcal{N}_0$,  
				$\Theta(t)=T(t)$  converges to $\Theta_{\mathcal{N}_0}=T_{\mathcal{N}_0}=\max\left(0,1-\dfrac{1}{\mathcal{R}_{\mathcal{N}_0}}\right)$.
				
				The disease persists if and only if $\mathcal{R}_{\mathcal{N}_0}>1$.
				
				\item $\mathcal{R}_{\mathcal{N}}>\mathcal{R}_{\mathcal{N}_0}$. 
				
				\item If $\mathcal{R}_{\mathcal{N}}>1$ then $\Theta_{\mathcal{N}_0}<\Theta_{\mathcal{N}}$ 
			\end{enumerate} 
		\end{lemma}
		\begin{proof}
			$(i)$ is a direct application of theorem \ref{th:SISqsk} with $\K\subset\mathbb{N}$, $q_k=\dfrac{kp_k}{\E_p(k)}$, $\beta_k=\lambda k$, $r_k=r$ and $\gamma_k=\gamma$.
			
			$(ii)$ is well known and may also be seen as a particular case of $(i)$ with only one class of host.  
			
			$(iii)$ is a consequence of the positivity of the variance :
			$\E_p(k^2)>\E_p(k)^2.$
			
			$(iv)$ If $\Theta_{\mathcal{N}_0}=0$ there is nothing to prove. Assume that $\Theta_{\mathcal{N}_0}>0$. For any $\Theta\in(0,1)$ define $g_\Theta(k)=\dfrac{k^2}{1+\rho k\Theta}$.
			
			By definition of $\Theta_{\mathcal{N}_0}$ and $\Theta_{\mathcal{N}}$, we have 
			
			\begin{equation}\label{Jensen}
				g_{\Theta_{\mathcal{N}_0}}(\E_p(k))=\dfrac{\E_p(k)}{\rho}=\E_p(g_{\Theta_{\mathcal{N}}}(k)).\end{equation}
			
			It is straightforward to show that $k\mapsto g_\Theta(k)$ is a  convex function for any $\Theta\in(0,1)$ and 
			the Jensen inequality implies
			\[\E_p(g_\Theta(k))>g_\Theta(\E_p(k)).\]
			
			Thus
			\[g_{\Theta_{\mathcal{N}_0}}(\E_p(k))=\E_p(g_{\Theta_{\mathcal{N}}}(k))>g_{\Theta_{\mathcal{N}}}(\E_p(k)\]
			and the conclusion follows by the decreasing nature  of $\Theta\mapsto g_\Theta(\E_p(k))$. 
			
		\end{proof}
		This lemma shows that keeping fixed the mean contact rate, the disease benefits from the heterogeneity of the network in two ways. 
        
        Firstly qualitatively (point $(iii)$ ) for the mere persistence of the disease, but also secondly quantitatively (point $(iv)$ ):  when the disease persists in both a homogeneous population $\mathcal{N}_0$ and a contact-structured population $\mathcal{N},$ then the global probability $\Theta$ to be infected is higher in the network case $\mathcal{N}$. \

\begin{rmk}
    It is instructive to contrast point~(iv), which is specific to a network structure, with the general result of Proposition~\ref{prop:SIDS}. 
Two different definitions of heterogeneity yield opposite effects on prevalence. 

In the general situation of Proposition~\ref{prop:SIDS}, the basic reproduction number $\cR_0$ is fixed, and we find that host heterogeneity reduces the prevalence. 
This finding is consistent with existing results in the literature and is referred to as the "frailty effect": 
individuals in highly susceptible classes are infected quickly, leading to a decline in the overall prevalence.
  
    In the random network situation of \eqref{lemma:net0}, the mean number of contacts is fixed. 
As noted, for instance, in \cite{moreno2002epidemic}, increasing heterogeneity in the degree distribution (e.g., higher variance of degree) 
enhances the likelihood of high‑degree nodes, which act as persistent transmission hubs and, in turn, increase the prevalence.

\end{rmk}

		Next, we focus on the interaction between strains.  In the particular situation when the strain perturbation appears in the $\alpha$'s (pairwise co-infection vulnerabilities among strains), the dynamics are driven by $\mu$ \cite{gjini2020key}. The following lemma shows that $\mu$ is always lower in a heterogeneous network than in a homogeneous network with the same mean connectivity.
		\begin{lemma}\label{mubar}
			Let $\mathcal{N}=(\K,(p_k)_{k\in\K})$ be any network with  mean connectivity  $\E_p(k)=\sum_{k\in\K} k p_k$ and $\mathcal{N}_0$ be the homogeneous network with constant connectivity $\E_p(k)$. Assume that $\mathcal{R}_{\mathcal{N}_0}=\rho \E_p(k)>1$.
			
			Set respectively $\mu(\mathcal{N})$ and $\mu(\mathcal{N}_0)$ the values of $\mu$ for these networks given by \eqref{munet}.
			
			We have 
			\[\mu(\mathcal{N})\leq\mu(\mathcal{N}_0)\]
			with equality if and only if $\mathcal{N}=\mathcal{N}_0$
		\end{lemma}
		\begin{proof}
			Assume that $\mathcal{N}\neq \mathcal{N}_0$ (otherwise there is nothing to prove).
			
			For the homogeneous network with constant connectivity $\E_p(k)$, the situation is the well known one class situation for which  we have ${\mathcal{R}_{\mathcal{N}_0}}=\rho \E_p(k)$, $\Theta_{\mathcal{N}_0}=T_{\mathcal{N}_0}=1-\dfrac{1}{\mathcal{R}_{\mathcal{N}_0}}$ and 
			\[\mu(\mathcal{N}_0)=\dfrac{1}{\sigma \mathcal{R}_{\mathcal{N}_0} \Theta_{\mathcal{N}_0}}=\dfrac{1}{\sigma (\rho{\E_p(k)}-1)}.\]

			We write \[\mu(\mathcal{N})=\mu(\mathcal{N}_0)\dfrac{\Theta_{\mathcal{N}_0}}{\Theta_{\mathcal{N}}}\sum_{k\in\K} h_k \dfrac{\E_p(k)}{k}\]
			
			Firstly, from the lemma \ref{lemma:net0}, we have
			$)\dfrac{\Theta_{\mathcal{N}_0}}{\Theta_{\mathcal{N}}}<1.$

			Secondly, define the function $f(k)=\dfrac{k^3}{(1+\sigma \rho k \Theta^*)^2 (1+\rho k \Theta^*)}$. The explicit formula \eqref{munet}   yields
			\[  \sum_{k\in\K} h_k \dfrac{\E_p(k)}{k}= \dfrac{\E_p(f(k))\E_p(k)}{\E_p(f(k)k)}.\]

			Since the function $f$ is increasing in $k$, we have:
			\[0<\E_p\left((k-\E_p(k))(f(k)-f\left(\E_p(k)\right))\right)
			=\E_p\left((k-\E_p(k))(f(k)-\E_p(f(k))\right)
			=\E_p\left(f(k)k\right)-\E_p(f(k))\E_p(k).\]

			Hence
			\[\sum_{k\in\K} h_k \dfrac{\E_p(k)}{k}= \dfrac{\E_p(f(k))\E_p(k)}{\E_p(f(k)k)}<1\]
			which ends the proof.

		\end{proof}

        \begin{figure}[t!]
			\centering
			\includegraphics[width=0.75\linewidth]{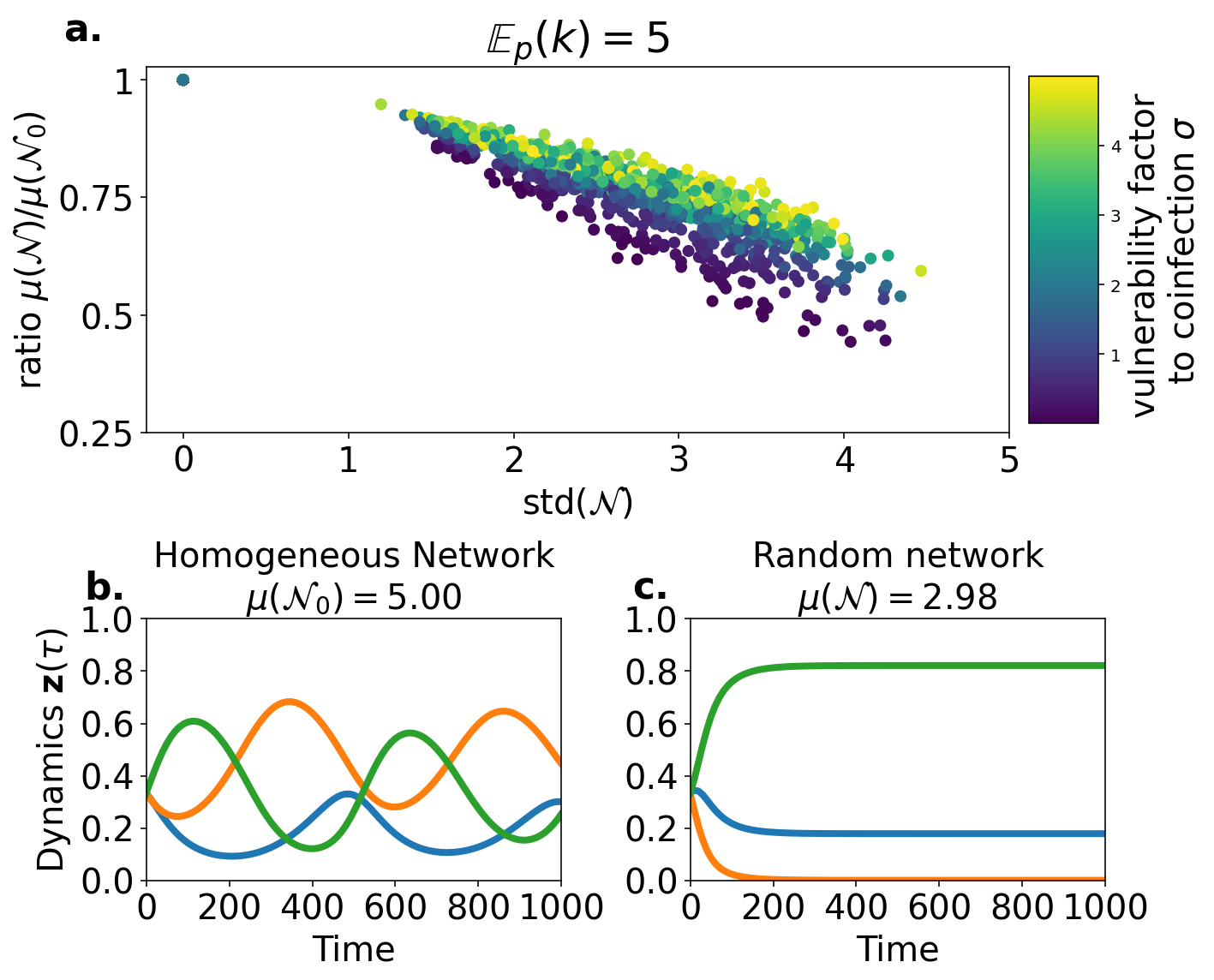}
            \caption{
            \textbf{Effect of host contact network heterogeneity on $\mu$ — a key driver of strain selection in the replicator.}\textbf{a.} We set and the per contact rate of infection $\rho=0.6$. We compare the values of $\mu(\mathcal{N})$ and $\mu(\mathcal{N}_0)$ for 1000 randomly generated networks and random coinfection susceptibility factor $\sigma\in[0,5]$. For each network we draw one marker with coordinates $\left(\mathrm{std}(\mathcal{N}), \mu(\mathcal{N})/\mu(\mathcal{N}_0)\right)$ and a color corresponding to the value of $\sigma$. We fix $\rho=0.6$  and the mean contact degree $\mathbb{E}_p(k)=5$ yielding  a fixed reproduction number in the homogeneous situation, $\mathcal{R}(\mathcal{N}_0)=\rho\,\mathbb{E}_p(k)=3.$.
            We observe that the ratio — and thus the effect of network heterogeneity on strain dynamics — increases with the heterogeneity $\mathrm{std}(\mathcal{N})$ and decreases with the epidemiological parameters $\sigma$. For a given deviation from neutrality $\alpha's$, a decrease in $\mu$ stabilizes the dynamics as shown in the two other sub-figures. \textbf{b.} Here we have set, $\sigma=0.4$ and $\rho=0.3$ thus $\mu(\mathcal{N}_0)=\left(\sigma (\rho\mathbb{E}_p(k)-1\right)^{-1}=5$ and 
$(\alpha_{ij})_{1\leq i,j\leq 3}=\begin{pmatrix} 0.15& -0.08&  0.49\\           -0.07& -0.16& -0.42\\
             0.32& -0.38& -0.37\end{pmatrix}$ and simulate the corresponding replicator system starting from uniform initial conditions among 3 strains. \textbf{c.} Same as in \textbf{b} but assuming a random contact network for the host population resulting in the new lower value of $\mu(\mathcal{N}).$}
	\label{fig:exampNetwork}
		\end{figure}	
		As a consequence, the host contact network may affect the strain dynamics by decreasing the value of $\mu$. We have previously shown that for strains varying only in $\alpha_{ij}$, lower values of $\mu$ keeping all else fixed, tend to create multi-stable coexistence fixed points with few strains, as opposed to large values of $\mu$ pushing the system towards more complex regimes of oscillatory attractors with many strains \cite{gjini2020key}. When strains vary in more dimensions, the $\mu$ effects on $\Lambda$ become nonlinear and even more complex \cite{le2023quasi}.   
		Now, how much $\mu$ will change in a network and its final effect will depend on both $\rho \E_p(k)$ (the overall strength of the infection), on $\sigma$ (the strength of co-infection relative to single infection) and on the heterogeneity of the network itself $\mathrm{std}(\mathcal{N})^2=V(\mathcal{N})=\E_p(k^2) - \E_p(k)^2$. Although we don't fully explore the downstream effects of $\mu$ on strain selection, Figure \ref{fig:exampNetwork}a illustrates precisely the altered values of $\mu$ in several heterogeneous host contact networks relative to the homogeneous population, expected to interfere subsequently with strain selection in the replicator equation (example in Fig. \ref{fig:exampNetwork}b). A more in-depth analysis of network effects on multi-strain dynamics is undertaken in \cite{networkPaper}, where we highlight how strain full interaction landscape also affects strongly the sensitivity of final selection dynamics to the network heterogeneity. Many such complexities in multi-strain systems with different layers and sources of host heterogeneity can now be approached with our framework.

        \section{Conclusion and discussion}
        \label{sec:discuss}
        Modelling multi-strain dynamics on realistic heterogeneous host populations is challenging because it combines three sources of complexity at once: pathogen diversity and interactions, host-level heterogeneity, and eventual contact network structure. The difficulties can be conceptual, mathematical, computational, and also data-related. Even single-strain epidemics on networks are hard to analyze exactly \cite{abubakar2025pairwise}. With multiple strains, 
pair and higher-order moment closures become unreliable, invasion criteria for new strains are unclear, stability and coexistence conditions are difficult to derive, and classic quantities like $\mathcal{R}_0$ become strain- and population structure-dependent. 
This limits theoretical insight and forces reliance on simulations and restricted parameter regimes, e.g., few strains \cite{yu2025spatial, zhang2016epidemic,pinotti_interplay_2020} or omission of important processes such as co-infection or detailed immune histories. Here, we bring a significant contribution to address these challenges \cite{wikramaratna2015five}, providing a new powerful analytical avenue for studying, simulating, and predicting multi-strain dynamics with coinfection, on structured host populations, including host contact networks.

        The core difficulty in multi-strain formalisms on networks is that strain interactions, host heterogeneity, and network structure amplify each other, typically leading to very large state spaces, loss of analytical simplicity, heavy computational demands, and limited empirical validation. 
        By using the strain similarity assumption as a simplifying core organizing principle in our epidemiological SIS multi-strain model with coinfection, we have shown, as in \cite{madec2020predicting,le2023quasi, madec2025derivation}, that model reduction is possible, arriving at a global replicator equation for strain frequencies over the entire host population (our central result in Theorem \ref{Th:replicator}). This has both computational and analytical advantages, including the promise of easier aggregated parameter identifiability and sensitivity analysis. Perhaps the most important conceptual and theoretical contribution of this finding is to link directly and mathematically mutual invasibility to co-circulation and coexistence dynamics between strains, in a significantly more general setup than previous work \cite{madec2020predicting, le2023quasi, le2022disentangling, park2024predicting}, explicitly capturing the roles of population structure and strain trait variability, and their interplay. 

Mathematically, our results are similar in spirit to our previous work invoking the replicator equation to reduce multi-strain epidemiological models \cite{madec2020predicting,le2023quasi,madec2025derivation}. 
As in these earlier studies, we apply Tikhonov’s theorem for slow–fast systems under the assumption of quasi-neutrality of strains. 
We thus rewrite the system in a slow–fast form, in which the fast dynamics are entirely determined by (i) the global aggregated variables, as in the single-strain case, and (ii) a linear system asymptotically governed by a matrix $A^*$ whose dominant eigenvalue is equal to $0$. 
The explicit characterization of this fast system allows us to derive the reduced slow dynamics, which necessarily take the form of a replicator equation.
Providing a complete description of this fast system is more delicate in the present setting than in our earlier work.

For (i), the global variables follow a non-trivial SIS-type system whose dynamics may {\em a priori} be challenging to describe. 
Here, (i) is resolved since the classical SIS system with host heterogeneity has been extensively studied in the literature, and the corresponding results extend directly to our coinfection model (Theorem \ref{th:SISqsk}). 
It is worth noting that in extensions of this framework—for instance, models with sequential clearance or infection rates depending on the current level of infection—Item~(i) may become a significantly more challenging problem.

On the other hand, Item~(ii) represents the main mathematical difficulty addressed in this paper, 
and the new approach we develop here can naturally be extended to more general coinfection models. This involves a matrix of size $2 \times p$, where $p>1$ denotes the number of host classes. 
In contrast to the trivial case $p=1$ considered in previous works, this requires the analysis of the spectrum of a matrix of potentially large dimension. 
We resolve this difficulty by rewriting the system in such a way that $A^*$ has a Metzler structure, which enables a detailed spectral analysis (lemma \ref{lemmaAstar}).

        While some of our results resonate with already-known features of epidemics in structured populations or networks \cite{van2002reproduction, pastor2001epidemic2}, many results are new, including: (i) the formal derivation of the model reduction based on similarity arguments and time-scale separation, (ii) the nontrivial way in which host population structure and strain selection intertwine in the explicit payoff matrix of the replicator equation, denoting pairwise invasion fitnesses between strains, (iii) the exact role of population heterogeneity on multi-strain endemic disease, and (iv) finally, the generality of the derived expressions for any population structure, intrinsic to the host population or induced by interventions.

Most notably, we show that heterogeneity can have opposing effects on different quantities of epidemiological interest, depending on what is controlled for. We formally prove that in a heterogeneous population, at endemic steady state, the overall prevalence of infection is lower than in a homogeneous population with the same $\mathcal{R}_0$ (Proposition \ref{prop:prevalence}). On the other hand, the threshold for disease persistence is lower in heterogeneous contact networks with same mean degree, and probability of infection is higher, making it easier for diseases to persist (Lemma \ref{lemma:net0}). We also prove that in heterogeneous host contact networks, the key quantity related to a special average of single-to-coinfection prevalence ratio across the entire population $(\mu)$ is always lower than its counterpart in a homogeneous setting with same mean connectivity (Lemma \ref{mubar}). However in more general heterogeneous population structures, the $\mu$ could be higher or lower than the one in the homogeneous setting (see Application 2 in Section 4, Figure \ref{fig:vaccinmu}), with potentially different effects on strain coexistence, e.g. stabilizing vs. destabilizing, or reducing the number of coexisting strains vs. increasing it \cite{gjini2020key}. 

This paper provides a basis upon which several theoretical and application extensions can be made in the fascinating and highly-relevant field of multi-strain and multi-host parasites \cite{webster2017acquires}. In particular, the special case of our SIS model with $N$ strains on a host contact network is studied in detail in a related paper \cite{networkPaper}. Future research building on this framework could investigate specific scenarios by focusing on particular host population structures and explicitly incorporating the effects of the mean, variance, and higher moments to clarify their qualitative and quantitative roles in shaping the dynamics. Further work could extend the examples presented here by accounting in greater detail for factors such as vaccination, age structure, and heterogeneity in susceptibility and infectivity \cite{mossong2008social,davies2020age,franco2022inferring,langwig2017vaccine}, within host populations or metapopulations. Finally, extensions of this model to account for stochasticity effects in small discrete populations remain an avenue of ongoing and prospective methodological developments. Applications can be far-ranging, from infectious diseases, e.g., multi-strain endemic respiratory pathogens, to multi-species colonization and co-colonization processes, opinion propagation in structured media, complex systems, and evolutionary dynamics. 
		
		\section{Proofs}
        \label{sec:proof}
	\subsection{Proofs of the Theorem \ref{Neutral:thm} on the neutral dynamics}\label{proofs:neutral}
    \subsubsection{A technical Lemma}
    
        Let us start by a technical lemma on irreducibility in block matrices.
		\begin{lemma}\label{lemma:irrecutiblekro}
			Let $Q=(q_{ks})_{1\leq k,s\leq p}$ be a non-negative matrix in $\mathcal{M}_p$ and $B_1,\cdots,B_p$ a family of $p$ strictly positive matrices in $\mathcal{M}_2$. Define
			\[M=\text{diag}(B_1,\cdots,B_p) \left(Q\otimes \mathbb{I}_2\right).\]
			In other words,   the square matrix $M=(M_{ks})_{k,s}\in\mathcal{M}_{pn}$ is  given by the $p\times p$ block structure
			\[M_{ks}=q_{ks} B_k\in\mathcal{M}_n.\]
			
			If $Q$ is irreducible, then $M$ is also irreducible.
		\end{lemma}
		\begin{rmk}
			In particular, if $B_k=B$ for each $k=1,\cdots,p$, then
			\[M=Q\otimes B.\]
		\end{rmk}
		\begin{proof}
			This is a direct consequence of the graph-theoretic characterization of irreducibility for nonnegative matrices (see, e.g., \cite{Bellman1968}). 
			Since each $B_k$ is strictly positive\footnote{
				The strict positivity of each $B_k$ is essential. If the $B_k$ are only assumed to be non-negative and irreducible, the resulting matrix $M$ may still be reducible. 
				For instance, taking 
				$Q=B_1=B_2=\begin{pmatrix}0&1\\1&0 \end{pmatrix}$ 
				produces the reducible matrix
				\[
				M=diag(B_1,B_2)(Q\otimes \mathbb{I}_2)
				=\begin{pmatrix}0&1\\1&0 \end{pmatrix}\otimes\begin{pmatrix}0&1\\1&0 \end{pmatrix}
				=\begin{pmatrix} 
					0&0&0&1\\
					0&0&1&0\\
					0&1&0&0\\
					1&0&0&0
				\end{pmatrix}.
				\]
			}, every edge $q_{ks}>0$ in the graph of $Q$ induces all possible edges between the corresponding block vertices in $M$. 
			Thus, the strong connectivity of $Q$ carries over to $M$, and $M$ is irreducible. 
			For related developments in the context of block matrices, see \cite{Romance2015}.
		\end{proof}
	\subsubsection{Proof of Lemma \ref{lemmaAstar}}\label{proofs:neutral33}
    (i) is immediate.  
			(ii) follows since only diagonal terms may be negative, and irreducibility of $Q\otimes_{col}(B_1,\ldots,B_p)$ is inherited from $Q$ (Lemma \ref{lemma:irrecutiblekro}).  \\
			(iii) Writing $X_k^*=(I_k^*,D_k^*)^T$, one checks directly that $A^*X^*=\bze_{2p}$ using the equilibrium equations.  \\
			(iv) By Perron–Frobenius theory for irreducible Metzler matrices, $0$ is a simple eigenvalue with positive eigenvector, all other eigenvalues having negative real part.  \\
			(v) The existence of a positive left eigenvector $\omega^*=(\phi_k^*,\psi_k^*)_{k\in\K}$ follows from (iv) by the Perron–Frobenius theorem for Metzler matrices. Since $\omega^*$ is unique up to a multiplicative constant, it is uniquely determined by the normalization $\omega^* X^*=1$.  
			
			The condition $\omega^*A^*=\bze_{2p}^T$ can be written explicitly as
			\begin{equation}\label{eq:omega}
				\forall s\in\K:\; (0,0)=\sum_{k\in\K} (\phi^*_k,\psi_k^*)A_{ks}^*
				=\left(\sum_{k\in\K}\beta_k\Bigl(S_k^*\phi^*_k+\tfrac{\sigma_k I_k^*}{2}\psi_k^*\Bigr)q_{ks}\right)(1,1)
				-(r_s+\gamma_s)(\phi^*_s,\psi_s^*)-\sigma_s\beta_s\Theta^*_s\Bigl(\phi_s^*-\tfrac12 \psi_s^*,0\Bigr).
			\end{equation}
			
			For each $s\in\K$, define
			\begin{equation}\label{eq:defus}
				u_s=\sum_{k\in\K}\beta_k\Bigl(S_k^*\phi^*_k+\tfrac{\sigma_k I_k^*}{2}\psi_k^*\Bigr)q_{ks}.
			\end{equation}
			Then, at each $s\in\K$, one obtains
			\[
			\begin{cases}
				(r_s+\gamma_s+\sigma_s\beta_s\Theta_s^*)\,\phi_s^*=u_s+\tfrac12\sigma_s\beta_s\Theta_s^*\,\psi_s^*,\\[1ex]
				(r_s+\gamma_s)\,\psi_s^*=u_s.
			\end{cases}
			\]
			Introducing
			\[
			\xi_s^*=\frac{\sigma_s\beta_s\Theta_s^*}{r_k+\gamma_k+\sigma_s\beta_s\Theta_s^*}=\sigma_s \dfrac{I_s^*}{S_s^*},
			\]
			we deduce
			\begin{equation}\label{eq:B4us}
				\phi_s^*=\frac{u_s}{r_s+\gamma_s}\Bigl(1-\tfrac12\xi_s^*\Bigr),\qquad
				\psi_s^*=\frac{u_s}{r_s+\gamma_s},\quad \forall s\in\K.
			\end{equation}
			
			It remains to compute $u=(u_s)_{s\in\K}$.
			Substituting \eqref{eq:B4us} into \eqref{eq:defus} yields
			\[
			u_s=\sum_{k\in\K}\mathcal{R}_k S_k^* u_k q_{ks}
			+\sum_{k\in\K}\frac{\mathcal{R}_k u_k}{2}\Bigl(\sigma_k I_k^*-S_k^*\xi_k^*\Bigr)q_{ks}.
			\]
			Since by definition $\sigma_k I_k^*-S_k^*\xi_k^*=0$ for all $k$, it follows that $u=(u_s)_{s\in\K}$ satisfies
			\[
			u^T=u^T\text{diag}(\mathcal{R}S^*)Q.
			\]
			By Corollary \ref{corxi}, one deduces $u=\mathcal{X}^*\pi$ for some scalar $\mathcal{X}^*>0$, which is uniquely determined by the normalization $\omega^*X^*=1$.\qed
	\subsubsection{Proof of Lemma \ref{lemma:Neutral}}
\label{proofs:neutral34}

\begin{proof}Denote $E=\R^2$ and $E^p\equiv \R^{2p}$. Let $F={\omega^*}^\perp:=\{Y\in E^p \mid \omega^* Y=0\}$. Then $E^p=\text{span}(X^*)\oplus F$.  
			Define $\Pi_F X=X-(\omega^* X) X^*$ as the projection onto $F$. and take $\mathfrak{X}^i=\Pi_FX^i$. We have $X^i=u^i X^* + \mathfrak{X}^i$.
			
			Moreover, both $\text{span}(X^*)$ and $F$ are invariant under $A^*$, which implies
			\[
			\Pi_F A^*=A^* \Pi_F.
			\]
			Hence, setting $u^i=\omega^* X^i$, we see that \eqref{eq:At} is equivalent to
			\[
			\begin{cases}
				\dfrac{d}{dt} u^i(t)=
                \omega^*\left(A(t)-A^*\right)X^*u_i(t)
                +\omega^*\left(A(t)-A^*\right) \mathfrak{X}^i(t),\\[6pt]
				\dfrac{d}{dt} \mathfrak{X}^i(t)(t)=A^*\mathfrak{X}^i(t)(t)+\Pi_F\left(A(t)-A^*\right) X^i(t).
			\end{cases}
			\]
			From the second equation, we have
            \[\mathfrak{X}^i(t)=e^{tA^*}\mathfrak{X}^i(0)+\int_0^t e^{(t-s)A^*} \Pi_F\left(A(s)-A^*\right) X^i(s). \]
            From the lemma \ref{lemmaAstar} -(iv), we know that  all eigenvalues of $\Pi_FA^*$, lie in the open left half-plane. This together with the assumption of the exponential convergence of $A(t)$ and the fact that  $X^i$ is bounded shows that  there exists $\eta_1>0$ and $C_1>0$ such that for any $t\geq 0$:
			\[
			\|\mathfrak{X}^i(t)\|\leq Ce^{-\eta_1 t} .
			\]
            
            From the first equation, we see  directly that there exists $C_2>0$ and $\eta_2$ such that 
            \begin{equation}\label{duiexp}
            \left|\dfrac{d}{dt} u^i(t) \right|\leq C_2 e^{-\eta_2 t}
            \end{equation}
            By the Cauchy criteria, if follows  that $\lim_{t\to+\infty}u^i(t):=z^i \geq0$ exists.

            Moreover, from the assumption $\sum_{i=1}^N X^i(t)\to X^*$ we get
            \[\sum_{i=1}^N u^i(t)=\omega^* \sum_{i=1}^N X^i(t)\to \omega^* X^*=1
            \]
            the last equality coming from the mere definition of $\omega^*$ in lemma \ref{lemmaAstar}.

            In other word, $\bz\in\Sigma^N$ and   \[
        \lim_{t\to+\infty}\mathbf{X}(t)=X^*\otimes\bz.
    \]
      \end{proof}

	\subsubsection{Proof of Theorem \ref{Neutral:thm}}
    \label{proofs:neutral35}
    \begin{proof}
        \begin{enumerate}[label=(\roman*)]
        \item 
A direct computation shows that $X^i$ satisfies
\[
\frac{d}{dt} X^i = A(t)X^i,
\]
where $A(t)$ depends on time only through the aggregated variables
$(S_k(t),I_k(t),D_k(t))_{k\in\K}$, whose behavior is described in
Proposition~\ref{Neutrall null 1 : aggregated variables.}
Therefore,
\[
(S_k(t),I_k(t),D_k(t))_{k\in\K}
\to (S_k^*,I_k^*,D_k^*)_{k\in\K}=:X^*.
\]
The result follows from Lemma~\ref{lemma:Neutral}.

     \item Again from Lemma~\ref{lemma:Neutral}, we have
$\mathbf{X}(t)\to X^*\otimes \bz$ as $t\to+\infty$, which yields,
for any $k\in\K$ and $i\in\lrbN$,
\[
\lim_{t\to+\infty} (I_k^i(t),D_k^i(t))
= (I_k^* z^i,\, D_k^* z^i).
\]
It follows that
\[
\lim_{t\to+\infty}\Theta^i(t)
= \lim_{t\to+\infty} Q(I^i(t)+D^i(t))
= Q(I^*+D^*)\, z^i.
\]
In other words, $\bTheta(t)\to \Theta^*\otimes \bz$ as $t\to+\infty$.

Now, the equation for $D_k^{ij}$ is
\[
\frac{d}{dt} D_k^{ij}
= \beta_k\sigma_k I_k^i \Theta_k^j - (r_k+\gamma_k) D_k^{ij},
\]
and a standard argument shows that
\[
D_k^{ij}\to
\frac{\beta_k\sigma_k \Theta_k^* I_k^*}{r_k+\gamma_k}\, z^i z^j
= D_k^* z^i z^j
\quad\text{as } t\to+\infty.
\]

        \item The invariance of $\mathcal{S}$ follows from a direct computation.
        \end{enumerate}
    \end{proof}
	\subsection{Proof of the Theorem \ref{Th:replicator} on the Quasi-Neutral dynamics
    \label{proofs:replicator}}
    
   Let $\varepsilon>0$ be fixed. Under the quasi-neutral assumption (strain similarity), the system
\eqref{mainsys} reads, for each $k\in\K$ (with the susceptible equation removed,
since the total population in each class is constant and equal to~1):
\begin{equation}\label{mainsys-eps}
\begin{cases}
\dfrac{d}{dt} I^{i,\varepsilon}_k
 = \beta_k S_k^\varepsilon \Theta_k^{i,\varepsilon}
   - I^{i,\varepsilon}_k \sum_{j=1}^N \beta_k\sigma_k \Theta_k^{j,\varepsilon}
   - (r_k+\gamma_k) I^{i,\varepsilon}_k
   + \varepsilon f_k^i(\bI_k^\varepsilon,\bD_k^\varepsilon,\bTheta_k^\varepsilon)
   + o(\varepsilon),\\[0.3em]
\dfrac{d}{dt} D_k^{ij,\varepsilon}
 = \beta_k\sigma_k \Theta_k^{j,\varepsilon} I^{i,\varepsilon}_k
   - (r_k+\gamma_k) D_k^{ij,\varepsilon}
   + \varepsilon \mathfrak{g}_k^{ij}(\bI_k^\varepsilon,\bD_k^\varepsilon,\bTheta_k^\varepsilon)
   + o(\varepsilon).
\end{cases}
\end{equation}

Here
\[
\Theta_k^{i,\varepsilon}
   = \sum_{s\in\K} q_{ks}(I_s^{i,\varepsilon} + D_s^{i,\varepsilon}),
   \qquad
D_s^{i,\varepsilon} =  \sum_{j\in\lrbN} \bigl(\mathbb{P}_k^{(i,j)\to i}D_s^{ij,\varepsilon} + \mathbb{P}_k^{(j,i)\to i}D_s^{ji,\varepsilon}\bigr),
\]
and the functions $s_k^i$, $f_k^i$, and $g_k^{ij}$ denote the first-order
expansion given by the quasi-neutral definition~\ref{def:qneutral}. Their
explicit expressions are provided in Appendix~\ref{appendix:order1}.

As in the perfectly neutral case ($\varepsilon=0$), the system can be rewritten
in an almost\footnote{The system is not strictly triangular because all
variables appear in $\bar f_k$ and $\bar g_k$. However, the structure is
perfectly triangular when $\varepsilon=0$, which is the key ingredient for the
slow–fast reduction.} triangular form:
(a) the aggregated quantities
$I_k^\varepsilon=\sum_{i\in\lrbN} I_k^{i,\varepsilon}$ and
$D_k^\varepsilon=\sum_{(i,j)\in\lrbN^2} D_k^{ij,\varepsilon}$, collected in
$X^\varepsilon=\bigl(I_k^\varepsilon,D_k^\varepsilon\bigr)_{k\in\K}$;
(b) the infectious variables
$X^{i,\varepsilon}=\bigl(I_k^{i,\varepsilon},D_k^{i,\varepsilon}\bigr)_{k\in\K}$;
and (c) the second-level infection variables $D_k^{ij,\varepsilon}$.
We obtain:
\begin{subequations}\label{qn:triangle}
\begin{align}
\dfrac{d}{dt} X^\varepsilon
 &= \begin{pmatrix}
    \beta_k S_k^\varepsilon \Theta_k^\varepsilon
      - \bigl(\beta_k\sigma_k\Theta_k^\varepsilon - (r_k+\gamma_k)\bigr) I_k^\varepsilon\\
    \beta_k\sigma_k \Theta_k^\varepsilon I_k^\varepsilon
      - (r_k+\gamma_k) D_k^\varepsilon
    \end{pmatrix}_{k\in\K}
    + \varepsilon
      \begin{pmatrix}
      \bar f_k(\bI_k^\varepsilon,\bD_k^\varepsilon,\bTheta_k^\varepsilon)\\
      \bar g_k(\bI_k^\varepsilon,\bD_k^\varepsilon,\bTheta_k^\varepsilon)
      \end{pmatrix}_{k\in\K}
    + o(\varepsilon),\label{qn:agreg}\\[0.4em]
\dfrac{d}{dt} X^{i,\varepsilon}
 &= A(X^\varepsilon) X^{i,\varepsilon}
    + \varepsilon\bigl(\mathcal{F}_k^i(\bI_k^\varepsilon,\bD_k^\varepsilon,\bTheta_k^\varepsilon)\bigr)_{k\in\K}
    + o(\varepsilon),\label{qn:sf}\\[0.4em]
\dfrac{d}{dt} D_k^{ij,\varepsilon}
 &= \beta_k\sigma_k \Theta_k^{j,\varepsilon} I_k^{i,\varepsilon}
    - (r_k+\gamma_k) D_k^{ij,\varepsilon}
    + \varepsilon g_k^{ij}(\bI_k^\varepsilon,\bD_k^\varepsilon,\bTheta_k^\varepsilon)
    + o(\varepsilon).\label{qn:Dij}
\end{align}
\end{subequations}

Here 
$
\mathcal{F}_k^i
 = \begin{pmatrix}
   f_k^i\\[0.2em]
   g_k^i\end{pmatrix}$,
$\bar f_k = \sum_{i\in\lrbN} f_k^i,$ and$
\bar g_k = \sum_{i\in\lrbN} g_k^{i}$
wherein we have set 
   \begin{equation}\label{eq:defFi}
g_k^i(\bI_k^\varepsilon,\bD_k^\varepsilon,\bTheta_k^\varepsilon)=\sum_{j\in\lrbN} \frac12(\mathfrak{g}_k^{ij}(\bI_k^\varepsilon,\bD_k^\varepsilon,\bTheta_k^\varepsilon)+\mathfrak{g}_k^{ji}(\bI_k^\varepsilon,\bD_k^\varepsilon,\bTheta_k^\varepsilon))+\beta_k\sigma_k \sum_{j=1}^N (\omega_k^{i,j}-\omega_k^{j,i})\Theta_k^{j,\eps} I_k^{i,\eps} ,\end{equation}

\medskip
\paragraph{Slow–fast decomposition.}
To apply Tikhonov’s theorem, we focus on~\eqref{qn:sf}.  
Introduce the decomposition
\[
u^{i,\varepsilon} = \omega^* X^{i,\varepsilon},
\qquad
\mathfrak{X}^{i,\varepsilon} = \Pi_F X^{i,\varepsilon},
\qquad
X^{i,\varepsilon} = u^{i,\varepsilon} X^* + \mathfrak{X}^{i,\varepsilon},
\]
which yields
\begin{equation}\label{qn:change}
\begin{cases}
\dfrac{d}{dt} u^{i,\varepsilon}
 = \omega^*\bigl(A(X^\varepsilon)-A^*\bigr)
   (u^{i,\varepsilon}X^*+\mathfrak{X}^{i,\varepsilon})
   + \varepsilon\,\omega^*\mathcal{F}^i(\bI^\varepsilon,\bD^\varepsilon,\bTheta^\varepsilon)
   + o(\varepsilon),\\[0.4em]
\dfrac{d}{dt} \mathfrak{X}^{i,\varepsilon}
 = \Pi_F A^* \mathfrak{X}^{i,\varepsilon}
   + \Pi_F\bigl(A(X^\varepsilon)-A^*\bigr)
     (u^{i,\varepsilon}X^*+\mathfrak{X}^{i,\varepsilon})
   + \varepsilon\,\Pi_F\mathcal{F}^i(\bI^\varepsilon,\bD^\varepsilon,\bTheta^\varepsilon)
   + o(\varepsilon).
\end{cases}
\end{equation}

All variables except $\bu^\varepsilon$ evolve on the fast $O(1)$ time scale,
while $\bu^\varepsilon$ evolves on the slow $O(\varepsilon)$ time scale.

\medskip
\paragraph{Fast limit.}
Letting $\varepsilon\to0$ in~\eqref{qn:triangle} yields the fast system
\begin{equation}\label{qnfast:triangle}
\begin{cases}
\dfrac{d}{dt} X = \begin{pmatrix}
\beta_k S_k\Theta_k - (\beta_k\sigma_k\Theta_k - (r_k+\gamma_k)) I_k\\
\beta_k\sigma_k\Theta_k I_k - (r_k+\gamma_k) D_k
\end{pmatrix}_{k\in\K},\\[0.4em]
\dfrac{d}{dt} \mathfrak{X}^i
 = \Pi_F A^* \mathfrak{X}^i
   + \Pi_F\bigl(A(I_\cdot,D_\cdot)-A^*\bigr)(z^i X^*+\mathfrak{X}^i),\\[0.4em]
\dfrac{d}{dt} D_k^{ij}
 = \beta_k\sigma_k \Theta_k^j I_k^i - (r_k+\gamma_k) D_k^{ij},\\[0.4em]
\dfrac{d}{dt} u^i
 = \omega^*\bigl(A(X)-A^*\bigr)(u^i X^*+\mathfrak{X}^i).
\end{cases}
\end{equation}

This is exactly the neutral system, fully described by
Theorems~\ref{th:SISqsk} and~\ref{Neutral:thm}. In particular,
\begin{equation}\label{fastlimit}
\lim_{t\to+\infty} \mathfrak{X}(t)=\mathbf{0},\qquad
\lim_{t\to+\infty} \bu(t)=\bz\in\Sigma^N,
\qquad
\lim_{t\to+\infty} (\bI(t),\bD(t),\bTheta(t))
 = (I^*\otimes\bz,\; D^*\otimes\bz\otimes\bz,\; \Theta^*\otimes\bz),
\end{equation}
with exponential convergence.

\medskip
\paragraph{Slow limit.}
Set $\tau=\varepsilon t$.  
From the first equation of~\eqref{qn:change} we obtain
\begin{equation}\label{qn:slow}
\dfrac{d}{d\tau} u^{i,\varepsilon}
 = \frac{1}{\varepsilon}\,
   \omega^*\bigl(A^\varepsilon(X^\varepsilon(\tfrac{\tau}{\varepsilon}))-A^*\bigr)
   (u^{i,\varepsilon}X^*+\mathfrak{X}^{i,\varepsilon})
 + \omega^*\mathcal{F}^i(\bI^\varepsilon(\tfrac{\tau}{\varepsilon}),
                         \bD^\varepsilon(\tfrac{\tau}{\varepsilon}),
                         \bTheta^\varepsilon(\tfrac{\tau}{\varepsilon}))
 + o(1).
\end{equation}

Fix $\tau>0$.

 By the quasi–neutrality assumption, the map
\[
\varepsilon \longmapsto A^\varepsilon(X)
\]
is $C^1$ at $\varepsilon=0$, uniformly for $X$ in a neighbourhood of the
fast attractor. Together with the exponential convergence of
$A^\varepsilon(X^\varepsilon(t))$ towards $A^*$ as $t\to+\infty$, this implies
that the composite map
\[
\varepsilon \longmapsto A^\varepsilon\!\left(X^\varepsilon\!\left(\tfrac{\tau}{\varepsilon}\right)\right)
\]
is differentiable at $\varepsilon=0$ for every fixed $\tau>0$. Hence the limit
\[
U(\tau):=\lim_{\varepsilon\to0}
\frac{A^\varepsilon(X^\varepsilon(\tfrac{\tau}{\varepsilon}))-A^*}{\varepsilon}
\]
exists and is simply the derivative at $\varepsilon=0$.

As $\varepsilon\to0$, applying \eqref{fastlimit} shows that
$\mathfrak{X}^{i,\varepsilon}(\tau/\varepsilon)\to0$ and
$\bu^\varepsilon\to\bz(\tau)\in\Sigma^N$, where $\bz(\tau)$ satisfies
\begin{equation}\label{qn:slowlim}
\dfrac{d}{d\tau} z^i
 = \bigl(\omega^* U(\tau) X^*\bigr) z^i
   + \omega^*\mathcal{F}^i(I^*\otimes\bz,\; D^*\otimes\bz\otimes\bz,\; \Theta^*\otimes\bz),
\end{equation}

A direct computation shows that $z^i$ factors out of $\mathcal{F}^i$:
\begin{equation}\label{def:fi}
\omega^*\mathcal{F}^i(I^*\otimes\bz,\; D^*\otimes\bz\otimes\bz,\; \Theta^*\otimes\bz)
 = z^i f^i(\bz),
\end{equation}
for an explicit functions $f^i$ given in \eqref{eq:definition_fi}.

Setting $-q(\tau)=\omega^*U(\tau)X^*$, we obtain
\begin{equation}\label{qn:slowlim2}
\dfrac{d}{d\tau} z^i(\tau)
 = z^i(\tau)\bigl(f^i(\bz(\tau)) - q(\tau)\bigr).
\end{equation}

Since $\bz\in\Sigma^N$, we have
$\sum_{i\in\lrbN} z^i(f^i(\bz)-q)=0$, hence
$q(\tau)=\sum_{i\in\lrbN} z^i f^i(\bz)$.
Thus the slow dynamics reduces to the replicator equation
\begin{equation}\label{qn:slowlim3}
\dfrac{d}{d\tau} z^i(\tau)
 = z^i(\tau)\Bigl(f^i(\bz(\tau))
   - \sum_{j\in\lrbN} z^j(\tau) f^j(\bz(\tau))\Bigr).
\end{equation}

To obtain the explicit form of the replicator coefficients, it suffices to
compute the functions $f^i$, defined by~\eqref{def:fi} from the first-order
expansion $\mathcal{F}_k^i$ of the quasi-neutral system.The details of these expansion is given in appendix \ref{appendix:order1}.
It appears that the functions $f^i$ are all linear in $\bz$, that is,
\[
f^i(\bz)=\sum_{j=1}^N m^{ij} z^j
\]
for some parameters $(m^{ij})$. Setting $M=(m^{ij})$, equation~\eqref{qn:slowlim3} takes the form of the replicator equation for $\bz\in \Sigma^N$:
\begin{equation}\label{qn:slowlim4}
\dfrac{d}{d\tau} z^i(\tau)
 = z^i(\tau)\Bigl((M\bz(\tau))_i
   - \bz(\tau)^T M\bz(\tau)\Bigr);\quad i=1,\cdots,N.
\end{equation}

There is no reason for $m^{ii}$ to vanish at this stage. However, it is classical that a linear replicator equation remains invariant under the addition of any constant to a column of $M$. Hence, setting $\Lambda=(\lambda_i^j)_{1\leq i,j\leq N}$ where
\[
\lambda_i^j = m^{ij} - m^{jj},
\]
we obtain the final formula
\begin{equation}\label{qn:slowlim5}
\dfrac{d}{d\tau} z^i(\tau)
 = z^i(\tau)\Bigl((\Lambda\bz(\tau))_i
   - \bz(\tau)^T \Lambda\bz(\tau)\Bigr),
\end{equation}
with $\Lambda$ explicitely given by the formula in Theorem~\ref{Th:replicator}.

{\bf Estimates}
The Tickonov's theorem implies that there exists a constant $C_1$ such that for any $\eps\in(0,\eps_0)$ and $\tau\in [\tau_0,T]$
\[\|X^{i,\eps}\left(\dfrac{\tau}{\eps}\right)-X^*z^i(\tau) \|\leq C_1 \eps.\]
The estimate on $(S_k,\bI_k,\bD_k)_{k\in\K}$ follows.

	
    \section*{Acknowledgements}
    This work received funding from the Portuguese Foundation for Science and Technology (FCT grant number 2022.03060.PTDC - Models4Invasion) and was partly supported by the European Commission (NOSEVAC-Modelling grant nr 101159175).
    \bibliography{references_Theoreticalpaper}
	\appendix
	\section{On Metzler Matrices}\label{Appendix:Metzler}
	A square matrix $\mathbb{A}\in\mathcal{M}_n$ is called a Metzler matrix if all its off-diagonal elements are non-negative. Thus, for a large enough $m\in\R$ the matrix $\mathbb{B}=\mathbb{A}+m\mathbb{I}_n$ is non-negative and we have $\lambda\in\text{sp}(\mathbb{A})\Longleftrightarrow\lambda+m\in\text{sp}(\mathbb{B})$.
	
	Therefore, the spectrum of a Metzler matrix inherits key characteristics from the spectrum of positive matrices, thanks to the Perron–Frobenius Theorem for Metzler matrices (adapted from Theorem 10.2 \cite{FB-LNS}).

    We recall the notation 
    \[\alpha(\mathbb{A})=\sup \{ \mathrm{Re}(\lambda),\;\lambda\in\mathrm{sp}(\mathbb{A})\}.\]
	\begin{thm}[Perron-Frobenius Theorem for Metzler matrix]\label{Th:A1}
		Let $n\in\N^*$ and $\mathbb{A}\in\mathcal{M}_n$ be Metzler.
        \begin{enumerate}[label=(\roman*)]
        \item $\alpha(\mathbb{A})$ is an  eigenvalue of $\mathbb{A}$, and
        \item the right and left eigenvector of $\alpha(\mathbb{A})$ are non-negative
        \end{enumerate}
        If additionally $\mathbb{A}$ is irreducible, then
        \begin{enumerate}[label=(\roman*),start=3]
        \item for any other eigenvalues $\lambda\in\mathrm{sp}(\mathbb{A})$, $\mathrm{Re}(\lambda)<\alpha(\mathbb{A})$, and
        \item the right and left eigenvectors of $\alpha(\mathbb{A})$ are unique and positive up to a multiplicative constant. 
        \end{enumerate}
	\end{thm}
	A lot of properties on Metzler Matrix are proven in \cite{FB-LNS}, see in particular Theorems 10.3 and 10.14. Here, we extract some properties we use in this article. purpose.
	\begin{thm}[Stability criteria for Metzler Matrices]\label{Th:A2}
    Let $\mathbb{A}\in\mathcal{M}_n$ be an irreducible Metzler matrix. The following statements are equivalent:
    \begin{enumerate}[label=(\roman*)]
        \item $\alpha(\mathbb{A})<0$,
        \item $\mathbb{A}$ is invertible and $-\mathbb{A}^{-1}>\mathbb{O}_n$,
        \item there exists $\mathbf{v}>0$ such that $\mathbb{A}\mathbf{v}<\mathbf{0}_n$,
        \item there exists $\zeta>\mathbf{0}_n$ such that $\zeta^T\mathbb{A}<\mathbf{0}_n^T$,
        \item for each $\mathbf{v}\geq \mathbf{0}_n$ and $\mathbf{v}\neq \mathbf{0}_n$,
        $\mathbb{A}\mathbf{v}$ has at least one negative entry.
        \end{enumerate}
	\end{thm}

	\section{Explicit first-order expansion functions}
\label{appendix:order1}

Using the expansion notations of Table~\ref{tabledef}, the function appearing in the first-order expansion of equation~\eqref{mainsys-eps} reads explicitly
\[
f_k^i(\bI_k,\bD_k,\bTheta_k)=
b_k^i \Theta_k^i S_k - c_k^i I_k^i
- \sum_{j=1}^N (\beta_k \alpha_{k}^{ij} + b_k^j \sigma_k)\, \Theta_k^j I_k^i .
\]
Similarly,
\[
\mathfrak{g}_k^{ij}(\bI_k,\bD_k,\bTheta_k)=
-c_k^{ij} D_k^{ij}
+ (\beta_k \alpha_{k}^{ij} + b_k^i \sigma_k)\, \Theta_k^j I_k^i .
\]

When rewriting the system in terms of \((I_k^i, D_k^i)_{k\in\K}\) in~\eqref{eq:defFi}, we obtain
\[
\mathcal{F}^i=\begin{pmatrix} f_k^i \\ g_k^i \end{pmatrix}_{k\in\K},
\]
where
\[
g_k^i(\bI_k,\bD_k,\bTheta_k)
= \sum_{j=1}^N \frac12\big(\mathfrak{g}_k^{ij}(\bI_k,\bD_k,\bTheta_k)
+ \mathfrak{g}_k^{ji}(\bI_k,\bD_k,\bTheta_k)\big)
+ \beta_k \sigma_k \sum_{j=1}^N (\omega_k^{i,j}-\omega_k^{j,i})\, \Theta_k^{j} I_k^{i},
\]
which expands explicitly as
\begin{align*}
g_k^i(\bI_k,\bD_k,\bTheta_k)=
&\;\frac12 \sigma_k \sum_{j=1}^N \big(b_k^i \Theta_k^j I_k^i + b_k^j \Theta_k^i I_k^j\big)\\
&\;-\sum_{j=1}^N \frac12 \big(c_k^{ji} D_k^{ji} + c_k^{ij} D_k^{ij}\big)\\
&\;+\frac12 \beta_k \sum_{j=1}^N \big(\alpha_k^{ij} \Theta_k^j I_k^i
+ \alpha_k^{ji} \Theta_k^i I_k^j\big)\\
&\;+\beta_k \sigma_k \sum_{j=1}^N (\omega_k^{i,j}-\omega_k^{j,i})\, \Theta_k^{j} I_k^{i}.
\end{align*}

From~\eqref{qn:slowlim}, it suffices to compute
\[
\omega^* \mathcal{F}^i(I^*\otimes\bz,\; D^*\otimes\bz\otimes\bz,\; \Theta^*\otimes\bz)
= \sum_{k\in\K} \big(\phi_k^* f_k^i + \psi_k^* g_k^i\big)
(I^*\otimes\bz,\; D^*\otimes\bz\otimes\bz,\; \Theta^*\otimes\bz),
\]
where \(\sum_{j=1}^N z^j = 1\) and \(\omega^*=(\phi_k^*,\psi_k^*)_{k\in\K}\).

We have
\[
f_k^i(I^*\otimes\bz,\; D^*\otimes\bz\otimes\bz,\; \Theta^*\otimes\bz)
= z^i \sum_{j=1}^N z^j \Big(
\Theta_k^*(S_k^* b_k^i - \sigma_k I_k^* b_k^j)
- I_k^* c_k^i
- \beta_k \Theta_k^* I_k^* \alpha_k^{ij}
\Big),
\]
and
\begin{align*}
g_k^i(I^*\otimes\bz,\; D^*\otimes\bz\otimes\bz,\; \Theta^*\otimes\bz)
= z^i \sum_{j=1}^N z^j \Bigg(
&\frac12 \sigma_k \Theta_k^* I_k^* (b_k^i + b_k^j)\\
&- D_k^* \frac{c_k^{ji} + c_k^{ij}}{2}\\
&+ \frac12 \beta_k \Theta_k^* I_k^* (\alpha_k^{ij} + \alpha_k^{ji})\\
&+ \beta_k \sigma_k \Theta_k^* I_k^* (\omega_k^{i,j}-\omega_k^{j,i})
\Bigg).
\end{align*}

Hence,
\[
\omega^* \mathcal{F}^i(I^*\otimes\bz,\; D^*\otimes\bz\otimes\bz,\; \Theta^*\otimes\bz)
= z^i f^i(\bz),
\]
where, using the relations
\(\psi_k^*=\mathcal{X}^* \dfrac{\pi_k}{r_k+\gamma_k}\),
\(\phi_k^*=(1-\tfrac12 \xi_k^*) \psi_k^*\),
and \(\xi_k^* S_k^*=\sigma_k I_k^*\) from Lemma~\ref{lemmaAstar},
\begin{equation}\label{eq:definition_fi}
f^i(\bz)=\sum_{j=1}^N m^{ij} z^j
= \mathcal{X}^* \sum_{k\in\K} \frac{\pi_k}{r_k+\gamma_k}
\sum_{j=1}^N \mathfrak{m}_k^{ij} z^j,
\end{equation}
with
\begin{align*}
\mathfrak{m}_k^{ij}=\;
&\Theta_k^* S_k^* b_k^i
- \frac12 \sigma_k \Theta_k^* I_k^* (1+\xi_k^*) b_k^j\\
&+ I_k^* \Big(1-\frac12 \xi_k^*\Big) c_k^i
+ D_k^* \frac{c_k^{ji}+c_k^{ij}}{2}\\
&+ \beta_k \sigma_k \Theta_k^* I_k^* (w_k^{ij}-w_k^{ji})\\
&+ \frac12 \beta_k \Theta_k^* I_k^*
\Big(\xi_k^* \alpha_k^{ji}
+ (1-\xi_k^*)(\alpha_k^{ji}-\alpha_k^{ij})\Big).
\end{align*}

The last term 

To obtain the final replicator equation~\eqref{qn:slowlim5} in terms of the fitness matrix \(\Lambda\), we define
\(\lambda_i^j = m^{ij} - m^{jj}\), which yields
\[
\lambda_i^j
= \mathcal{X}^* \sum_{k\in\K} \frac{\pi_k}{r_k+\gamma_k} Y_k^{i,j},
\]
 where \(Y_k^{ij}=m_k^{ij}-m_k^{jj}\), satisfy 
\begin{align*}Y_k^{ij}=\;
&\Theta_k^* S_k^* (b_k^i-b_k^j)
+ I_k^* \Big(1-\frac12 \xi_k^*\Big) (c_k^i-c_k^j)
+ D_k^* \left(\frac{c_k^{ji}+c_k^{ij}}{2}-c_k^{jj}\right)+ \beta_k \sigma_k \Theta_k^* I_k^* (w_k^{ij}-w_k^{ji})\\
&+ \frac12 \beta_k \Theta_k^* I_k^*
\Big(\xi_k^* (\alpha_k^{ji}-\alpha_k^{jj})
+ (1-\xi_k^*)(\alpha_k^{ji}-\alpha_k^{ij})\Big).
\end{align*}

 Factorizing the last term by $\xi_k^*$ and denoting $\mu_k=\dfrac{1-\xi_k^*}{\xi_k^*}=\dfrac{I_k^*}{D_k^*}$ gives  the final expression~\eqref{lambdaindetail}.
    \end{document}